\documentclass[12pt]{article}
\usepackage{pifont}
\usepackage{jheppub}
\usepackage{graphicx}
\usepackage{float}
\usepackage{epsfig}
\usepackage{amsfonts}
\usepackage{amssymb}
\usepackage{amsmath}
\usepackage{amsthm}
\usepackage{subfig}
\usepackage{bm}
\usepackage{hyperref}
\usepackage{mathrsfs}
\usepackage{bbm}
\usepackage{cancel}
\usepackage{color}
\usepackage{accents}

\def\j{\alpha_j}

\def\be{\begin{equation}}
\def\ee{\end{equation}}
\def\bea{\begin{eqnarray}}
\def\eea{\end{eqnarray}}
\newcommand{\beas}{\begin{eqnarray*}}
\newcommand{\eeas}{\end{eqnarray*}}

\newcommand{\nn}{\nonumber \\}
\newcommand{\nb}{\nonumber}

\newtheorem{theorem}{\sf THEOREM}
\newtheorem{definition}{Definition}[section]
\newtheorem{lemma}[theorem]{Lemma}

\newtheorem{prop}[theorem]{Proposition}

\def\abs#1{\left| #1\right|}

\def\Tr{\mathop{\rm Tr}}

\def\Sl{\sum\limits}
\def\Label#1{\label{#1}%
  \smash{\hbox to0pt{\raise1ex\hbox{\tiny[#1]}\hss}}}

\allowdisplaybreaks
\title{Off-shell BCJ Relation in Nonlinear Sigma Model}

\author[a,b]{Gang Chen,}
\author[c]{Shuyi Li}
\author[d]{and Hanqing Liu}

\affiliation[a]{Department of Physics, Zhejiang Normal University 688 Yingbin Road, Jinhua 321004, China}
\affiliation[b]{Department of Physics, Nanjing University 22 Hankou Road, Nanjing 210093, China}
\affiliation[c]{Department of Physics, Brown University, Providence, RI 02906, USA}
\affiliation[d]{Department of Physics, Duke University, Durham, NC 27708, USA}
\emailAdd{gang.chern@gmail.com}
\emailAdd{shuyi\_li@brown.edu}
\emailAdd{hanqing.liu@duke.edu}

\date{\today}
\abstract{We investigate relations among tree-level off-shell currents in nonlinear sigma model. Under Cayley parametrization, we propose and prove a general revised BCJ relation for even-point currents. Unlike the on-shell BCJ relation, the off-shell one behaves quite differently from Yang-Mills theory although the algebraic structure is the same. After performing the permutation summation in the revised BCJ relation, the sum is non-vanishing, instead, it equals to the sum of sub-current products with the BCJ coefficients under a specific ordering, which is presented by an explicit formula. Taking on-shell limit, this identity is reduced to the on-shell BCJ relation, and thus provides the full off-shell correspondence of tree-level BCJ relation in nonlinear sigma model.}


\begin{document}
\maketitle
\section{Introduction}
Discovering new amplitude relations is one of the significant tasks in scattering amplitudes in recent years. A celebrated inspection that there is a duality between color factors and kinetic factors in Yang-Mills theory \cite{Bern:2008qj} made by Bern, Carrasco and Johansson implies relations between color-ordered amplitudes at tree level. In \cite{Bern:2008qj}, the authors pointed out that the scattering amplitudes in Yang-Mills theory can be expressed by Feynman-like diagrams with only cubic vertices and thus the duality is established. Corresponding to the algebraic properties of the color factors, namely antisymmetry and Jacobi identity, we have KK relation \cite{Kleiss:1988ne} and BCJ relation \cite{Bern:2008qj} respectively. With KK relation, the number of independent color-ordered $n$-point tree-level amplitudes is reduced to $(n-2)!$, while with BCJ relation, it can be further reduced to $(n-3)!$.

The tree-level amplitude relations in Yang-Mills theory have been studied in both string theory \cite{BjerrumBohr:2009rd,Stieberger:2009hq} and field theory. In field theory, KK relation was first proven by new color decompositions \cite{DelDuca:1999rs}, then both KK and BCJ relations was proven by BCFW recursion \cite{Britto:2004ap, Britto:2005fq,Feng:2010my, Tye:2010kg, Cachazo:2012uq, Chen:2011jxa}. The kinematic factors in Yang-Mills theory can be constructed from pure-spinor string theory \cite{Mafra:2011kj}, from an area-preserving diffeomorphism algebra \cite{Monteiro:2011pc,BjerrumBohr:2012mg} or a more general diffeomorphism algebra \cite{Fu:2012uy}. They can also be understood through the construction of color-dual decomposition and trace-like objects \cite{Bern:2011ia,Du:2013sha,Fu:2013qna}, relabeling symmetry \cite{Broedel:2011pd,Fu:2014pya,Naculich:2014rta} and scattering equations \cite{Cachazo:2013gna, Cachazo:2013hca, Cachazo:2013iea,Naculich:2014rta,Naculich:2014naa}.

Many generalization of KK relation and BCJ relation has also been made. One direction is to study this duality at loop level \cite{Bern:2010yg, BjerrumBohr:2011xe, Carrasco:2011mn, Boels:2011tp,Boels:2011mn,  Carrasco:2012ca, Bern:2012uf, Du:2012mt, Oxburgh:2012zr, Saotome:2012vy, Boels:2012ew, Carrasco:2012ca, Boels:2013bi, Bjerrum-Bohr:2013iza, Bern:2013yya,  Nohle:2013bfa}. Another direction is to see whether these relations have analogues in other theories. It is not surprising that they hold in most Yang-Mills-like theories, such as gauge theory coupled with matter \cite{Sondergaard:2009za}, $\mathcal{N}=4$ super Yang-Mills theory \cite{Jia:2010nz} and color-scalar amplitudes \cite{Du:2011js}, whose color factors share the same algebraic properties. While the algebraic properties change, as is the case of ABJM theory with 3-algebra \cite{Bargheer:2012gv}, the amplitude relations also change. Besides those fundamental theories, it is worthwhile to investigate some effective theories with the same  algebra as Yang-Mills theory, such as the nonlinear sigma model with $SU(N)$, which is governed by the chiral Lagrangian and describes the phenomenological behavior of the Goldstone bosons under the chiral symmetry breaking $SU(N)_L\times SU(N)_R\rightarrow SU(N)$. There are many progress in the amplitude of  nonlinear sigma model recently \cite{Kampf:2012fn, Kampf:2013vha, Chen:2013fya, Cachazo:2014xea, Chen:2014dfa, Du:2015esa, Carrasco:2016ldy, Du:2016tbc}.

In \cite{Ma:2011um}, it was pointed out that in field theory, all the on-shell KK relations and on-shell general BCJ relations can be generated through two primary relations: the fundamental BCJ relation and cyclic symmetry. Since nonlinear sigma model at tree level satisfies both primary relations, the on-shell KK and general BCJ relations, which are exactly the same as the counterparts in Yang-Mills theory, are also guaranteed \cite{Chen:2013fya}. Based on this result the fundamental BCJ relation with one external line off-shell is also proven \cite{Chen:2013fya}. It no longer shares the same formulae as in Yang-Mills theory, however.

In this paper, continuing the proof of the KK relation \cite{Chen:2014dfa}, we will further prove the general BCJ relation in tree-level nonlinear sigma model with one external line off-shell, which is closely based on the previous results \cite{Chen:2013fya,Chen:2014dfa}. We use Berends-Giele recursion relation under Cayley parametrization to calculate each diagram. As is shown in \cite{Kampf:2012fn,Kampf:2013vha}, all odd-point currents have to vanish, therefore we only need to consider even-point currents $\mathcal{J}(\sigma)$ in nonlinear sigma model. We conjecture and prove the general BCJ relation in this case. 

This paper is organized as following. We first provide the Feynman rules and Berends-Giele recursion relation in nonlinear sigma model in Section \ref{PPP}. Then we introduce the revised BCJ relation and prove that it is equivalent to the general BCJ relation in Section \ref{Sec:RBCJ}. In Section \ref{EX}, we verify the revised BCJ relation by calculate a six point example directly. Finally in Section \ref{PF}, we prove the revised BCJ relation by Berends-Giele recursion relation. 

\section{Preparation: Feynman rules, Berends-Giele recursion and $U(1)$ identity}\label{PPP}
In this section, we review the Feynman rules and the Berends-Giele recursion in nonlinear sigma model which are useful through this paper.\footnote{There is an overlap with the section 2 of \cite{Chen:2013fya, Chen:2014dfa}.} Most of the notations follow the recent papers \cite{Kampf:2012fn,Kampf:2013vha}. 

%%%%%%%%%%%%%%%%%%%%%
\subsection{Feynman rules}
%%%%%%%%%%%%%%%%%%%%%
\noindent {\textbf {Lagrangian}}

\noindent The Lagrangian of $U(N)$ nonlinear sigma model is
\bea
\mathcal{L}={F^2\over 4}\Tr (\partial_{\mu}U\partial^{\mu}U^{\dagger}),
\eea
where $F$ is a constant. As shown in \cite{Kampf:2012fn,Kampf:2013vha}, $U$ is defined by
\bea
U=1+2\Sl_{n=1}^{\infty}\left({1\over 2F}\phi\right)^n,~~~~\label{Cayley}
\eea
where $\phi=\sqrt{2}\phi^at^a$ and $t^a$ are generators of $U(N)$ Lie algebra.

\noindent{\textbf {Trace form of color decomposition}}

\noindent The full tree amplitudes can be given in terms of color-ordered amplitudes by trace form of color decomposition
\bea
M(1^{a_1},\dots,n^{a_n})=\Sl_{\sigma\in S_{n-1}}\Tr(T^{a_{1}}T^{a_{\sigma_2}}\dots T^{a_{\sigma_n}})A(1,\sigma).\label{Trace form}
\eea
Since the traces have cyclic symmetry, the color-ordered amplitudes also satisfy cyclic symmetry
\bea
A(1,2,\dots,n)=A(n,1,\dots,n-1).\label{Cyclic symmetry}
\eea

\noindent \textbf{Feynman rules for color-ordered amplitudes}

\noindent Vertices in color-ordered Feynman rules under Cayley parametrization \eqref{Cayley} are
\bea
V_{2n+1}=0, V_{2n+2}=\left(-{1\over 2F^2}\right)^n\left(\Sl_{i=0}^np_{2i+1}\right)^2=\left(-{1\over 2F^2}\right)^n\left(\Sl_{i=0}^np_{2i+2}\right)^2,\Label{Feyn-rules}
\eea
where $p_j$ denotes the momentum of the leg $j$. Momentum conservation has been considered.

%%%%%%%%%%%%%%%%%%%%%%%%
\subsection{Berends-Giele recursion}
%%%%%%%%%%%%%%%%%%%%%%%%
In the Feynman rules given above, one can construct tree-level currents\footnote{In this paper, an $n$-point current is mentioned as the current with $n-1$ on-shell legs and one off-shell leg.} with one off-shell line through Berends-Giele recursion
\bea
&&J(2,...,n)\nn
&=&\frac{i}{P_{2,n}^2}\Sl_{m=4}^n\Sl_{1=j_0<j_1<\cdots<j_{m-1}=n}i V_{m}(p_1=-P_{2,n},P_{j_0+1,j_1},\cdots,P_{j_{m-2}+1,n})\times\prod\limits_{k=0}^{m-2} J(j_k+1,\cdots,j_{k+1}),\Label{B-G}\nn
\eea
where $p_1=-P_{2,n}=-(p_2+p_3+\dots+p_n)$. The starting point of this recursion is $J(2)=J(3)=\dots=J(n)=1$.

\noindent There is at least one odd-point vertex for current with odd-point lines (including the off-shell line) and the odd-point vertices are zero. As a result, we have
\bea
J(2,\dots,2m+1)=0,
\eea
for $(2m+1)$-point amplitudes.
The currents with even points in general are nonzero and are built up by only odd numbers of even-point sub-currents. Since odd-point currents have to vanish,
in all following sections of this paper, we just need to discuss on the relations among even-point currents.

\section{Off-shell General BCJ relation and Revised BCJ relation}\label{Sec:RBCJ}
The on-shell case of general BCJ relation, which is guaranteed as is explained in the introduction, has the following form,
\begin{equation}\label{}
  \Sl_{\sigma\in OP(\mathbf{\alpha_r}\bigcup\mathbf{\beta_s})}\left(\Sl_{i=1}^r\Sl_{\xi_{\sigma_k}<\xi_{\alpha_i}}s_{\alpha_i\sigma_k}\right)\mathcal{A}(1,\{\sigma\},2m)=0,\footnote{$\mathbf{\alpha_r}$ is an abbreviation for $\alpha_1,\cdots,\alpha_r$.}\Label{on-shell-gen-BCJ}
\end{equation}
where $OP(\mathbf{\alpha_r}\bigcup\mathbf{\beta_s})$ means ordered permutation among two sets $\mathbf{\alpha_r}$ and $\mathbf{\beta_s}$, and $\xi_{\alpha_i}$ denotes the position of $\alpha_i$ in the ordered permutation $\sigma$.
As for the off-shell case, the value of this expression, which can be conveniently denoted by $\mathcal{A}_{BCJ}(r,s)$, is no longer zero.

In order to prove the off-shell general BCJ relation, it is convenient to introduce a new relation which is equivalent to the general BCJ relation.
\begin{theorem}
(Revised BCJ relation)
\begin{equation}\label{RBCJ}
  \begin{split}
  \mathcal{A}_{RBCJ}(r,s)\equiv&\Sl_{\sigma\in OP(\mathbf{\alpha_r}\bigcup\mathbf{\beta_s})}\left(\Sl_{i=1}^r\Sl_{\xi_{\sigma_k}<\xi_{\alpha_i}}s_{\alpha_i\sigma_k}\right)\mathcal{A}(1,\{\sigma\})\\
  =&\Sl_{div\{\mathbf{\alpha_r},\mathbf{\beta_s}\}}   \left(\frac{1}{2F^2}\right)^{\frac{R+S-1}{2}}p_1^2S_{div\{\mathbf{\alpha_r},\mathbf{\beta_s}\}}\mathcal{J}(A_{1})\cdots \mathcal{J}(A_R)\mathcal{J}(B_{1})\cdots \mathcal{J}(B_S),
  \end{split}
\end{equation}
where $\{A_1, \cdots, A_R, B_1,\cdots,B_S\}$is an element in $div\{\mathbf{\alpha_r},\mathbf{\beta_s}\}$. The coefficients $S_{div\{\mathbf{\alpha_r},\mathbf{\beta_s}\}}$ is a function of ${div\{\mathbf{\alpha_r},\mathbf{\beta_s}\}}$, and can be decided by the following diagrams Fig.\ref{S_{div}}.
\begin{figure}[htbp]
  \centering
  \subfloat[$R-S=1$]{
  \label{Fig_R-S=1}
    \begin{minipage}[t]{0.5\textwidth}
      \centering
      \includegraphics[width=6cm]{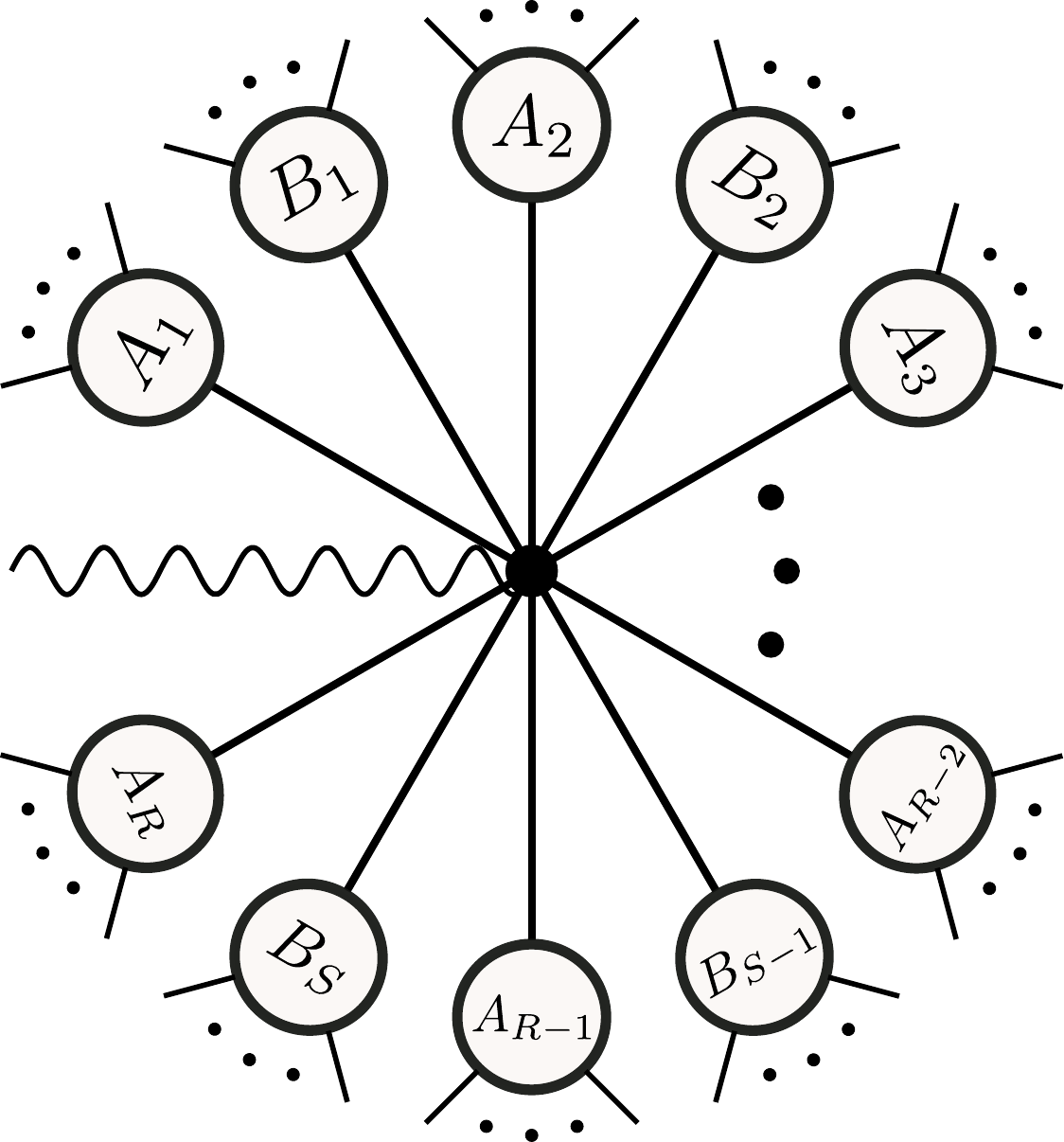}
    \end{minipage}
  }
  \subfloat[$R-S=-1$]{
  \label{Fig_R-S=-1}
    \begin{minipage}[t]{0.5\textwidth}
      \centering
      \includegraphics[width=6cm]{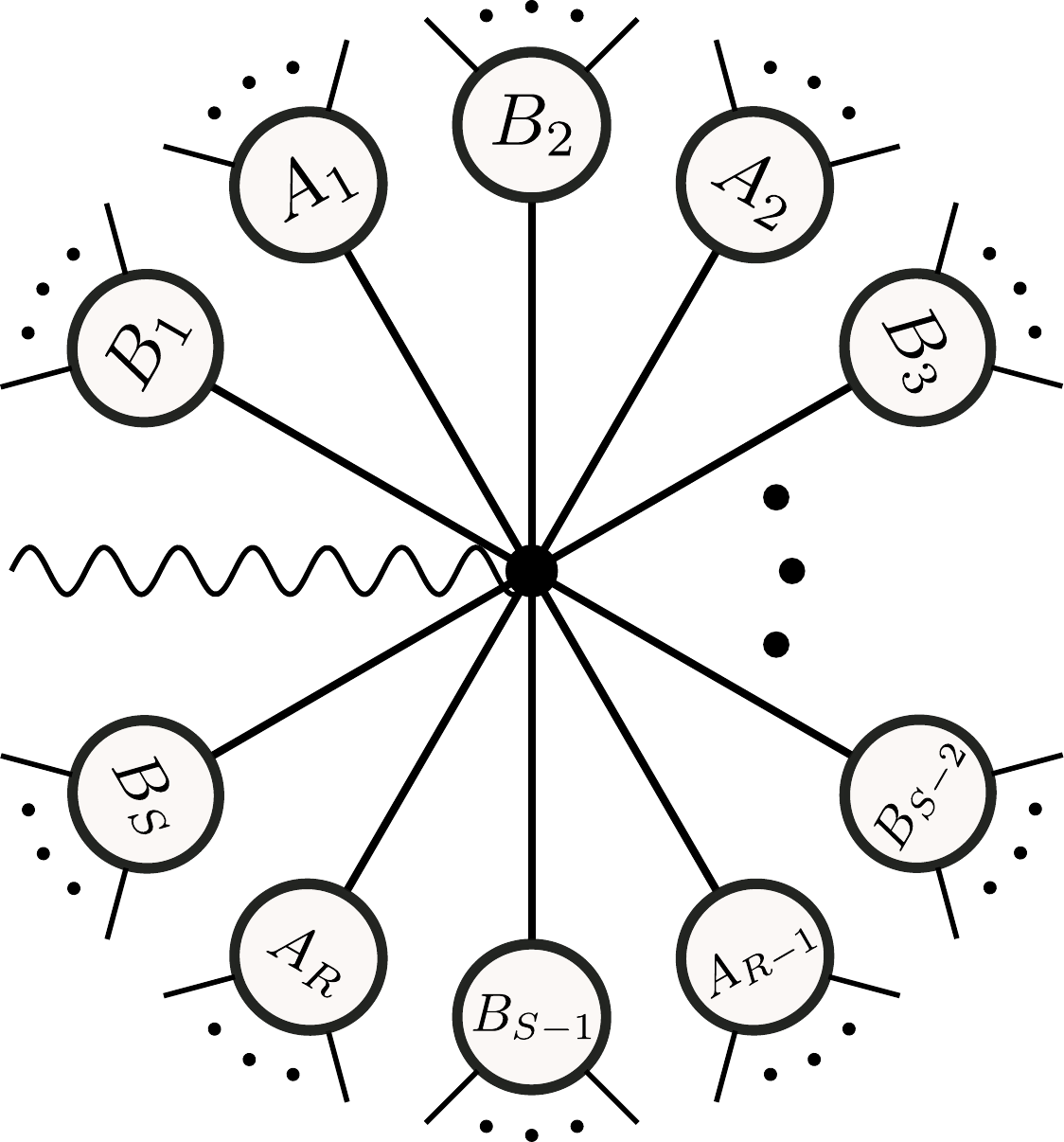}
  \end{minipage}
  }
\caption{Schematic diagram representations to read the coefficients  $S_{div\{\mathbf{\alpha_r},\mathbf{\beta_s}\}}$.}
\label{S_{div}}
\end{figure}

The rule to read the coefficients is the same as the general BCJ relation, that is, for every $\alpha_i$ and for every term $\sigma_k$ before $\alpha_i$ we have a coefficient $s_{\alpha_i\sigma_k}$.
\end{theorem}

Take the on-shell limit $p_1^2\rightarrow0$, and we arrive at the on-shell general BCJ relation (\ref{on-shell-gen-BCJ}).

Equivalently, we can use the condition of momentum conservation to change coefficients of both sides from $\Sl_{\xi_{\sigma_k}<\xi_{\alpha_i}}s_{\alpha_i\sigma_k}$ to $-\Sl_{\xi_{\sigma_k}>\xi_{\alpha_i}}s_{\alpha_i\sigma_k}$\footnote{We will keep this convention for simplicity in the rest of this paper.}. Our next task is to prove the equivalence between the revised BCJ relation and the general BCJ relation.

\begin{prop}
The revised BCJ relation is necessary and sufficient to the general BCJ relation.
\end{prop}
\begin{proof}
On the one hand, we can divide the revised BCJ relation into two parts according to the last term in the permutation.
\begin{equation}\label{RBCG-GBCG}
  \begin{split}
  &\Sl_{\sigma\in OP(\mathbf{\alpha_r}\bigcup\mathbf{\beta_s})}\left(\Sl_{i=1}^r\Sl_{\xi_{\sigma_k}<\xi_{\alpha_i}}s_{\alpha_i\sigma_k}\right)\mathcal{A}(1,\{\sigma\})\\
  =&\Sl_{\sigma\in OP(\mathbf{\alpha_{r-1}}\bigcup\mathbf{\beta_s})}\left(\Sl_{l=1}^{r-1}\Sl_{\xi_{\sigma_k}<\xi_{\alpha_i}}s_{\alpha_i\sigma_k}\right)\mathcal{A}(1,\{\sigma\},\alpha_r)\\
  &+\Sl_{\sigma\in OP(\mathbf{\alpha_r}\bigcup\mathbf{\beta_{s-1}})}\left(\Sl_{i=1}^r\Sl_{\xi_{\sigma_k}<\xi_{\alpha_i}}s_{\alpha_i\sigma_k}\right)\mathcal{A}(1,\{\sigma\},\beta_s).
  \end{split}
\end{equation}
or
\begin{equation*}
  \mathcal{A}_{RBCJ}(r,s)=\mathcal{A}_{BCJ}(r-1,s)+\mathcal{A}_{BCJ}(r,s-1)
\end{equation*}
by our notation. Thus if we know the expression of $\mathcal{A}_{BCJ}(r,s)$, we can get the expression of $\mathcal{A}_{RBCJ}(r,s)$ immediately.

On the other hand, we can use (\ref{RBCG-GBCG}) recursively,
\begin{equation*}\label{}
  \begin{split}
  \mathcal{A}_{BCJ}(r,s)=&\mathcal{A}_{RBCJ}(r,s+1)-\mathcal{A}_{BCJ}(r-1,s+1)\\
  \mathcal{A}_{BCJ}(r-1,s+1)=&\mathcal{A}_{RBCJ}(r-1,s+2)-\mathcal{A}_{BCJ}(r-2,s+2)\\
  \mathcal{A}_{BCJ}(r-2,s+2)=&\mathcal{A}_{RBCJ}(r-2,s+3)-\mathcal{A}_{BCJ}(r-3,s+3)\\
  \cdots=&\cdots\\
  \mathcal{A}_{BCJ}(2,s+r-2)=&\mathcal{A}_{RBCJ}(2,s+r-1)-\mathcal{A}_{BCJ}(1,s+r-1)\\
  \end{split}
\end{equation*}
Therefore we finally arrive at
\begin{equation}\label{}
  \mathcal{A}_{BCJ}(r,s)=\Sl_{i=0}^{r-2}(-1)^{i}\mathcal{A}_{RBCJ}(r-i,s+1+i)+(-1)^{r-1}\mathcal{A}_{BCJ}(1,r+s-1).
\end{equation}
Notice that $\mathcal{A}_{BCJ}(1,r+s-1)=\mathcal{A}_{RBCJ}(1,r+s)$, and we can write it uniformly,
\begin{equation}\label{}
  \mathcal{A}_{BCJ}(r,s)=\Sl_{i=0}^{r-1}(-1)^{i}\mathcal{A}_{RBCJ}(r-i,s+1+i).
\end{equation}

Since we have already proven the fundamental BCJ relation in \cite{Chen:2013fya}, i.e. we have already known the expression of $\mathcal{A}_{BCJ}(1,r+s-1)$, we can calculate $\mathcal{A}_{BCJ}(r,s)$ from $\mathcal{A}_{RBCJ}(r,s)$ simultaneously.\footnote{The notation of $\alpha$ and $\beta$ in $\mathcal{A}_{RBCJ}(r,s)$ should be treated very cautiously, and a good way is denoting $\mathbf{\alpha_r}$ by $\{\alpha_1\beta_{s+r}\cdots\beta_{s+2}\beta_{s+1}\}$, which follows the convention in the fundamental BCJ relation.}
\end{proof}
Now we have proven the equivalence, and thus, it is sufficient to prove the revised BCJ relation.

\paragraph{Remark}It is easy to see that this result is consistent with the result of the generalized $U(1)$-decoupling identity \cite{Chen:2014dfa}, by which we mean the generalized $U(1)$-decoupling identity can be easily verified by the revised BCJ relation. For each diagram, we have
\begin{equation}\label{}
    \mathcal{D}_{RBCJ}(r,s)+\mathcal{D}_{RBCJ}(r,\dot{s})=-p_1^2\mathcal{D}_{GU}(r,s),\footnote{The dot in $(r,\dot{s})$ means the BCJ coefficients are caused by $\beta$ rather than $\alpha$.}
\end{equation}
where $\mathcal{D}$ refers to a single diagram. Sum them up and we can get
\begin{equation}\label{BCJ-GU(1)}
    \mathcal{A}_{RBCJ}(r,s)+\mathcal{A}_{RBCJ}(r,\dot{s})=-p_1^2{\mathcal{A}_{GU}}(r,s).
\end{equation}

\section{Revised BCJ Relation for Six-point Currents}\label{EX}
Before proving the general result in Eq. \ref{RBCJ}, we first investigate a six point example. We have two external lines in $\mathbf{\alpha_r}$ and three lines in $\mathbf{\beta_s}$. The amplitude is denoted by $\mathcal{A}(2,3)$. The revised BCJ relation is
\begin{equation*}
\begin{aligned}
  \mathcal{A}_{RBCJ}(2,3)=&-\frac{1}{2F^2}p_1^2s_{\alpha_1, \alpha_2\beta_1\beta_2\beta_3}\mathcal{J}(\alpha_1)\mathcal{J}(\alpha_2)\mathcal{J}(\beta_1\beta_2\beta_3)\\
  &-\left(\frac{1}{2F^2}\right)^2p_1^2(s_{\alpha_1,\alpha_2\beta_2\beta_3}+s_{\alpha_2\beta_3})\mathcal{J}(\alpha_1)\mathcal{J}(\alpha_2)\mathcal{J}(\beta_1)\mathcal{J}(\beta_2)\mathcal{J}(\beta_3),
\end{aligned}\footnote{Generally, $2p_{\alpha_1}(\Sl_{i=m_1}^{m_2}p_{\alpha_i}+\Sl_{j=n_1}^{n_2} p_{\beta_j})$ is denoted by $s_{\alpha_1,\alpha_{m_1} \cdots \alpha_{m_2} \beta_{n_1} \cdots\beta_{n_2}}$.}
\end{equation*}
We will prove it by calculating it directly. In our calculation, let $\mathcal{J}(\alpha_1)\mathcal{J}(\alpha_2)\mathcal{J}(\beta_1)\mathcal{J}(\beta_2)\mathcal{J}(\beta_3)=1$. The diagrams of $\mathcal{A}_{RBCJ}(2,3)$ can be divided into two cases, one vertex and two vertices. And the two vertices case can be divided into three parts, as is presented in Fig.\ref{example} after the appendix.

The first part $A$ contains the sub-current $\mathcal{J}(\beta_1\beta_2\beta_3)$,
\begin{eqnarray*}
  \mathcal{A}_A&=&-p_1^2\left(-\frac{1}{p_1^2}\right)\left(-\frac{1}{2F^2}\right)\left[\right.(s_{\alpha_1,\alpha_2\beta_1\beta_2\beta_3}+s_{\alpha_2,\beta_1\beta_2\beta_3})(p_{\alpha_1}+p_{\beta_1}+p_{\beta_2}+p_{\beta_3})^2\\
  &+&\left.s_{\alpha_1,\alpha_2\beta_1\beta_2\beta_3}(p_{\alpha_1}+p_{\alpha_2})^2
  +s_{\alpha_1\alpha_2}(p_{\alpha_2}+p_{\beta_1}+p_{\beta_2}+p_{\beta_3})^2\right]\mathcal{J}(\alpha_1)\mathcal{J}(\alpha_2)\mathcal{J}(\beta_1\beta_2\beta_3)\\
  &=&-p_1^2\frac{1}{2F^2}s_{\alpha_1,\alpha_2\beta_1\beta_2\beta_3}\mathcal{J}(\alpha_1)\mathcal{J}(\alpha_2)\mathcal{J}(\beta_1\beta_2\beta_3)\\
  &=&-\left(\frac{1}{2F^2}\right)^2(s_{\alpha_1\alpha_2}+s_{\alpha_2,\beta_1\beta_2\beta_3})s_{\beta_1\beta_3}
\end{eqnarray*}

The second part $B$ contains the sub-currents which contain one of the $\mathbf{\alpha_r}$. This part gets contribution from four diagrams (B.1-B.4) in Fig.\ref{example},  $$\mathcal{A}_B=\mathcal{A}_{B_1}+\mathcal{A}_{B_2}+\mathcal{A}_{B_3}+\mathcal{A}_{B_4},$$
where 
\begin{equation*}
\begin{aligned}
  \mathcal{A}_{B_1}=&-p_1^2\left(-\frac{1}{p_1^2}\right)\left(-\frac{1}{2F^2}\right)\left\{(p_{\alpha_1}+p_{\beta_1}+p_{\beta_2}+p_{\beta_3})^2\right.\left[(s_{\alpha_1,\alpha_2\beta_1\beta_2\beta_3}+s_{\alpha_2,\beta_3})\mathcal{J}(\alpha_1\beta_1\beta_2)\right.\\
  &+(s_{\alpha_1,\alpha_2\beta_2\beta_3}+s_{\alpha_2,\beta_3})\mathcal{J}(\beta_1\alpha_1\beta_2)+(s_{\alpha_1,\alpha_2\beta_3}+s_{\alpha_2\beta_3})\mathcal{J}(\beta_1\beta_2\alpha_1)\left.\right]\\
  &+(p_{\alpha_1}+p_{\beta_1}+p_{\beta_2}+p_{\alpha_2})^2\left[\right.s_{\alpha_1,\alpha_2\beta_1\beta_2\beta_3}\mathcal{J}(\alpha_1\beta_1\beta_2)\\
  &+s_{\alpha_1,\alpha_2\beta_2\beta_3}\mathcal{J}(\beta_1\alpha_1\beta_2)+s_{\alpha_1,\alpha_2\beta_3}\mathcal{J}(\beta_1\beta_2\alpha_1)\left.\right]\left.\right\}\\
  =&-\left(\frac{1}{2F^2}\right)^2\left[(p_{\alpha_1}+p_{\beta_1}+p_{\beta_2}+p_{\beta_3})^2(s_{\alpha_1,\alpha_2\beta_2\beta_3}+s_{\alpha_2\beta_3})+(p_{\alpha_1}+p_{\beta_1}+p_{\beta_2}+p_{\alpha_2})^2s_{\alpha_1,\alpha_2\beta_2\beta_3}\right]
  \end{aligned}
\end{equation*}
\begin{equation*}
\begin{aligned}
  \mathcal{A}_{B_2}=&-\left(\frac{1}{2F^2}\right)^2(p_{\alpha_2}+p_{\beta_1})^2s_{\alpha_1,\alpha_2\beta_3}\\
  \mathcal{A}_{B_3}=&-\left(\frac{1}{2F^2}\right)^2\left[\right.(p_{\beta_1}+p_{\alpha_2}+p_{\beta_2}+p_{\beta_3})^2(s_{\alpha_1,\alpha_2\beta_2\beta_3}\\
  &+s_{\alpha_2\beta_3})+(p_{\alpha_1}+p_{\alpha_2}+p_{\beta_2}+p_{\beta_3})^2(s_{\alpha_1,\alpha_2\beta_1\beta_2\beta_3}+s_{\alpha_2\beta_3})\left.\right]\\
  \mathcal{A}_{B_4}=&-\left(\frac{1}{2F^2}\right)^2(p_{\alpha_1}+p_{\beta_3})^2(s_{\alpha_1,\alpha_2\beta_1\beta_2\beta_3}+s_{\alpha_2,\beta_2\beta_3}).
\end{aligned}
\end{equation*}

The third part $C$ contains the sub-currents which contain one of the $\mathbf{\beta_s}$ as shown in diagram (C.1-C.3) in Fig.\ref{example},
$$\mathcal{A}_C=\mathcal{A}_{C_1}+\mathcal{A}_{C_2}+\mathcal{A}_{C_3},$$
where
\begin{equation*}
\begin{aligned}
  \mathcal{A}_{C_1}=&-p_1^2\left(-\frac{1}{p_1^2}\right)\left(-\frac{1}{2F^2}\right)(p_{\alpha_1}+p_{\alpha_2}+p_{\beta_1}+p_{\beta_3})^2\left[\right.(s_{\alpha_1,\alpha_2\beta_1\beta_2\beta_3}+s_{\alpha_2,\beta_1\beta_2\beta_3})\mathcal{J}(\alpha_1\alpha_2\beta_1)\\
  &+(s_{\alpha_1,\alpha_2\beta_1\beta_2\beta_3}+s_{\alpha_2,\beta_2\beta_3})\mathcal{J}(\alpha_1\beta_1\alpha_2)+(s_{\alpha_1,\alpha_2\beta_2\beta_3}+s_{\alpha_2,\beta_2\beta_3})\mathcal{J}(\beta_1\alpha_1\alpha_2)\left.\right]\\
  =&-\left(\frac{1}{2F^2}\right)^2(p_{\alpha_1}+p_{\alpha_2}+p_{\beta_1}+p_{\beta_3})^2(s_{\alpha_1,\alpha_2\beta_1\beta_2\beta_3}+s_{\alpha_2,\beta_2\beta_3})\\
  \mathcal{A}_{C_2}=&-\left(\frac{1}{2F^2}\right)^2(p_{\beta_1}+p_{\beta_3})^2(s_{\alpha_1,\alpha_2\beta_2\beta_3}+s_{\alpha_2\beta_3})\\
  \mathcal{A}_{C_3}=&-\left(\frac{1}{2F^2}\right)^2(p_{\alpha_1}+p_{\alpha_2}+p_{\beta_1}+p_{\beta_3})^2s_{\alpha_1,\alpha_2\beta_3}.
\end{aligned}
\end{equation*}

The one vertex case $D$ is easy to obtain,
\begin{equation*}
\begin{aligned}
  \mathcal{A}_D=&-p_1^2\left(-\frac{1}{p_1^2}\right)\left(-\frac{1}{2F^2}\right)^2\left[\right.(p_{\alpha_1}+p_{\beta_1}+p_{\beta_3})^2(s_{\alpha_1,\alpha_2\beta_1\beta_2\beta_3}+s_{\alpha_2,\beta_1\beta_2\beta_3})\nb\\
  &+(p_{\alpha_1}+p_{\beta_2}+p_{\beta_3})^2(s_{\alpha_1,\alpha_2\beta_1\beta_2\beta_3}+s_{\alpha_2,\beta_3})+(p_{\alpha_1}+p_{\beta_2}+p_{\alpha_2})^2s_{\alpha_1,\alpha_2\beta_1\beta_2\beta_3}\\
  &+(p_{\beta_1}+p_{\alpha_2}+p_{\beta_3})^2(s_{\alpha_1,\alpha_2\beta_2\beta_3}+s_{\alpha_2,\beta_2\beta_3})+(p_{\beta_1}+p_{\beta_2}+p_{\beta_3})^2(s_{\alpha_1,\alpha_2\beta_2\beta_3}+s_{\alpha_2,\beta_3})\\
  &+(p_{\beta_1}+p_{\beta_2}+p_{\alpha_2})^2s_{\alpha_1,\alpha_2\beta_2\beta_3}+(p_{\beta_1}+p_{\alpha_1}+p_{\beta_3})^2(s_{\alpha_1,\alpha_2\beta_3}+s_{\alpha_2\beta_3})\\
  &+(p_{\beta_1}+p_{\alpha_1}+p_{\alpha_2})^2s_{\alpha_1,\alpha_2\beta_3}+(p_{\beta_1}+p_{\beta_3}+p_{\alpha_2})^2s_{\alpha_1\alpha_2}\\
  &+(p_{\alpha_1}+p_{\alpha_2}+p_{\beta_3})^2(s_{\alpha_1,\alpha_2\beta_1\beta_2\beta_3}+s_{\alpha_2,\beta_2\beta_3})\left.\right].
\end{aligned}
\end{equation*}
Considering all four parts, we arrive at the final result,
\begin{equation*}
\begin{aligned}
  \mathcal{A}_A+\mathcal{A}_B+\mathcal{A}_C+\mathcal{A}_D=&-p_1^2\left[\right.\frac{1}{2F^2}s_{\alpha_1,\alpha_2\beta_1\beta_2\beta_3}\mathcal{J}(\alpha_1)\mathcal{J}(\alpha_2)\mathcal{J}(\beta_1\beta_2\beta_3)\\
  &+\left(\frac{1}{2F^2}\right)^2(s_{\alpha_1,\alpha_2\beta_2\beta_3}+s_{\alpha_2\beta_3})\mathcal{J}(\alpha_1)\mathcal{J}(\alpha_2)\mathcal{J}(\beta_1)\mathcal{J}(\beta_2)\mathcal{J}(\beta_3)\left.\right],
\end{aligned}
\end{equation*}
where we have recovered the current $\mathcal{J}(\alpha_1)\mathcal{J}(\alpha_2)\mathcal{J}(\beta_1)\mathcal{J}(\beta_2)\mathcal{J}(\beta_3)$. It is easy to see that this result is consistent with the revised BCJ relation in Eq. \ref{RBCJ}. 

\section{The Proof of Revised Off-shell BCJ Relation}\label{PF}
We will prove (\ref{RBCJ}) recursively, which means when we calculate $\mathcal{A}_{RBCJ}(r,s)$, we will assume that we have already known the form of $\mathcal{A}_{RBCJ}(r',s')$ in the following conditions,
\begin{itemize}
  \item $r'+s'<r+s$,
  \item $r'+s'=r+s$, but $r'<r$.
\end{itemize}

Sometimes, we will use the result in sub-current form, that is, to divide both sides of (\ref{RBCJ}) by $p_1^2$,
\begin{equation}\label{RBCJ-J}
  \begin{split}
  &\Sl_{\sigma\in OP(\mathbf{\alpha_r}\bigcup\mathbf{\beta_s})}\left(\Sl_{i=1}^r\Sl_{\xi_{\sigma_k}>\xi_{\alpha_i}}-s_{\alpha_i\sigma_k}\right)\mathcal{J}(1,\{\sigma\})\\
  =&\Sl_{div\{\mathbf{\alpha_r},\mathbf{\beta_s}\}}   \left(\frac{1}{2F^2}\right)^{\frac{R+S-1}{2}}S_{div\{\mathbf{\alpha_r},\mathbf{\beta_s}\}}\mathcal{J}(A_{1})\cdots \mathcal{J}(A_R)\mathcal{J}(B_{1})\cdots \mathcal{J}(B_S).
  \end{split}
\end{equation}

Before we go ahead to the proof, let's first sketch it as following. 
\begin{enumerate}
\item In subsection \ref{sub_div_bcj}, we divide the BCJ coefficients into two parts, one of which can be obtained immediately from the KK relation \cite{Chen:2014dfa}, so we only need to consider the other part;
\item In subsection \ref{subsec:QP}, we prove that after the revised BCJ permutation summation, only quartic polynomials of external momenta would appear in the coefficients of each division;
\item In subsection \ref{subsec:Complet}, in order to make the argument more rigorous, we associate a totally ordered path to each contribution to a specific division, so that we know that the argument is complete and not repetitive;
\item In subsection \ref{subsec:CurrentRec}, we reduce the proof to the calculation of the coefficients of the divisions without sub-currents, and give the general form of it;
\item In subsection \ref{subsec:COM}, we calculate the coefficient matrices $\mathbf{C}(r,s)$. %The most part of this calculation, i.e. the matrices of $\mathbf{D}(r,s)$, is left in the appendix.
\end{enumerate}
\subsection{Dividing the BCJ coefficients}\label{sub_div_bcj}
Now let us divide $\mathcal{A}_{RBCJ}(r,s)$ into two parts, one contains the terms with coefficients $s_{\alpha_i\alpha_j}$, the other contains the terms with coefficients $s_{\alpha_i\beta_j}$. The first part is denoted by $\mathcal{A}_{RBCJ}^\mathrm{I}(r,s)$, and the second by $\mathcal{A}_{RBCJ}^\mathrm{II}(r,s)$. In the first part, the sum of coefficients $s_{\alpha_i\alpha_j}$ is the same in every term, so we can get the result of the first part by using the generalized $U(1)$-decoupling identity. Then we have,
\begin{equation}\label{}
  \mathcal{A}_{RBCJ}^\mathrm{I}(r,s)=P_{\mathbf{\alpha_r}}^2\Sl_{\sigma\in OP(\mathbf{\alpha_r}\bigcup\mathbf{\beta_s})}\mathcal{A}(1,\{\sigma\})
  =P_{\mathbf{\alpha_r}}^2{\mathcal{A}_{GU}}(r,s)
\end{equation}
Correspondingly, we have a recursion hypothesis in $\mathcal{A}_{RBCJ}^\mathrm{II}(r,s)$ form,
\begin{equation}\label{RBCJ-J}
  \begin{split}
  &\Sl_{\sigma\in OP(\mathbf{\alpha_r}\bigcup\mathbf{\beta_s})}\left(\Sl_{i=1}^r\Sl_{\xi_{\beta_j}>\xi_{\alpha_i}}-s_{\alpha_i\beta_j}\right)\mathcal{J}(1,\{\sigma\})\\
  =&\Sl_{div\{\mathbf{\alpha_r},\mathbf{\beta_s}\}}   \left(\frac{1}{2F^2}\right)^{\frac{R+S-1}{2}}S_{div\{\mathbf{\alpha_r},\mathbf{\beta_s}\}}^\mathrm{II}\mathcal{J}(A_{1})\cdots \mathcal{J}(A_R)\mathcal{J}(B_{1})\cdots \mathcal{J}(B_S).
  \end{split}
\end{equation}
where $S_{div\{\mathbf{\alpha_r},\mathbf{\beta_s}\}}^\mathrm{II}$ contains only $s_{\alpha_i\beta_j}$.

\subsection{Quartic polynomials in the coefficients of each division}\label{subsec:QP}
We can calculate $\mathcal{A}_{RBCJ}^\mathrm{II}(r,s)$ directly to demonstrate the following lemma,
\begin{lemma}\label{polynomial}
After we only permute the $\alpha$ and $\beta$ in same sub-currents attached directly to the vertices containing the off-shell line $1$, the coefficients of each division are all quartic polynomials of external momenta.
\end{lemma}
\begin{proof}
First, we can use Berends-Giele recursion relation to make every term to be summed explicit,
\begin{equation}\label{}
  \begin{split}
  \mathcal{A}_{RBCJ}^\mathrm{II}(r,s)=&\Sl_{\sigma\in OP(\mathbf{\alpha_r}\bigcup\mathbf{\beta_s})}\left(\Sl_{i=1}^r\Sl_{\xi_{\beta_j}>\xi_{\alpha_i}}-s_{\alpha_i\beta_j}\right)\mathcal{A}(1,\{\sigma\})\\
  =&\Sl_{\sigma\in OP(\mathbf{\alpha_r}\bigcup\mathbf{\beta_s})} \left(\Sl_{i=1}^r\Sl_{\xi_{\beta_j}>\xi_{\alpha_i}}s_{\alpha_i\beta_j}\right)\left(\Sl_{m=4}^n\Sl_{div}V(m,div,\sigma)\mathcal{J}_1\mathcal{J}_2\cdots \mathcal{J}_{m-1}\right),\\
  \end{split}
\end{equation}
where $\mathcal{J}_i$ means some sub-currents which are attached to the vertices containing the off-shell line $1$. Changing the order of the summations $\Sl_{\sigma\in OP(\mathbf{\alpha_r}\bigcup\mathbf{\beta_s})}$ and $\Sl_{m=4}^n\Sl_{div}$, we have
\begin{equation}\label{}
  \mathcal{A}_{RBCJ}^\mathrm{II}(r,s)=\Sl_{m=4}^n\Sl_{div}\Sl_{\sigma\in OP(\mathbf{\alpha_r}\bigcup\mathbf{\beta_s})} \left(\Sl_{i=1}^r\Sl_{\xi_{\beta_j}>\xi_{\alpha_i}}s_{\alpha_i\beta_j}\right)V(m,div,\sigma)\mathcal{J}_1\mathcal{J}_2\cdots \mathcal{J}_{m-1}
\end{equation}
Second, the permutation $\sigma$ can be divided into permutation $\sigma_k$ within the sub-current $\mathcal{J}_k$ and $\bar{\sigma}$ for the permutation involving exchange among different sub-currents, i.e. $\sigma=\bar{\sigma}\prod\limits_{k=1}^{m-1}\sigma_k$. Then we can split $\Sl_{\sigma\in OP(\mathbf{\alpha_r}\bigcup\mathbf{\beta_s})}$ into $\Sl_{\bar{\sigma}}\Sl_{\sigma_{m-1}}\Sl_{\sigma_{m-2}}\cdots\Sl_{\sigma_1}$,
\begin{equation}\label{}
  \mathcal{A}_{RBCJ}^\mathrm{II}(r,s)=\Sl_{m=4}^n\Sl_{div}\Sl_{\bar{\sigma}}\Sl_{\sigma_{m-1}}\Sl_{\sigma_{m-2}}\cdots\Sl_{\sigma_1}V(m,div,\sigma) \left(\Sl_{i=1}^r\Sl_{\xi_{\beta_j}>\xi_{\alpha_i}}s_{\alpha_i\beta_j}\right)\mathcal{J}_1\mathcal{J}_2\cdots \mathcal{J}_{m-1}.
\end{equation}
Correspondingly, we can divide the BCJ coefficients $\Sl_{i=1}^r\Sl_{\xi_{\beta_j}>\xi_{\alpha_i}}s_{\alpha_i\beta_j}$ into two parts. In the first part $S_k$, both indices in each $s_{\alpha_i\beta_j}$ belong to $\mathcal{J}_k$, while in the second part $\bar{S}$, two indices belong to different sub-currents. Therefore we have $\Sl_{i=1}^r\Sl_{\xi_{\beta_j}>\xi_{\alpha_i}}s_{\alpha_i\beta_j}=\Sl_{k=1}^{m-1}S_k+\bar{S}$, where
\begin{equation}\label{div_S}
  S_k=\Sl_{\alpha_i\in \mathcal{J}_k}\Sl_{\substack{\beta_j\in \mathcal{J}_k\\\xi_{\beta_j}>\xi_{\alpha_i}}}s_{\alpha_i\beta_j},\ \bar{S}=\Sl_{k=1}^{m-1}\Sl_{\alpha_i\in \mathcal{J}_k}\Sl_{\substack{\beta_j\notin \mathcal{J}_k\\\xi_{\beta_j}>\xi_{\alpha_i}}}s_{\alpha_i\beta_j}
\end{equation}
And since the expression of $V(m,div,\bar{\sigma}\prod\limits_{k=1}^{m-1}\sigma_k)$ is not changed under $\sigma_k$, it can be simply written as $V(m,div,\bar{\sigma})$. Having noticed that $\sigma_k$ only acts on $\mathcal{J}_k$, we have
\begin{equation}\label{$A_{RBCJ2}$-splitted}
  \begin{split}
  \mathcal{A}_{RBCJ}^\mathrm{II}(r,s)=&~~~\Sl_{m=4}^n\Sl_{div}\Sl_{\bar{\sigma}}V(m,div,\bar{\sigma})\left(\Sl_{\sigma_1} S_1\mathcal{J}_1\right) \left(\Sl_{\sigma_{2}}\mathcal{J}_{2}\right)\cdots\left(\Sl_{\sigma_{m-1}}\mathcal{J}_{m-1}\right)\\
  &+\Sl_{m=4}^n\Sl_{div}\Sl_{\bar{\sigma}}V(m,div,\bar{\sigma}) \left(\Sl_{\sigma_{1}}\mathcal{J}_{1}\right)\left(\Sl_{\sigma_{2}}S_2\mathcal{J}_2\right)\cdots\left(\Sl_{\sigma_{m-1}}\mathcal{J}_{m-1}\right)\\
  &~~~~~~~~~~~~~~~~~~~~~~~~~~~~~~~~~~~~~~\cdots\\
  &+\Sl_{m=4}^n\Sl_{div}\Sl_{\bar{\sigma}}V(m,div,\bar{\sigma})\left(\Sl_{\sigma_{1}}\mathcal{J}_{1}\right)\left(\Sl_{\sigma_{2}}\mathcal{J}_{2}\right)\cdots \left(\Sl_{\sigma_{m-1}}S_{m-1}\mathcal{J}_{m-1}\right)\\
  &+\Sl_{m=4}^n\Sl_{div}\Sl_{\bar{\sigma}}V(m,div,\bar{\sigma})\bar{S}\left(\Sl_{\sigma_{1}}\mathcal{J}_{1}\right)\left(\Sl_{\sigma_{2}}\mathcal{J}_{2}\right) \cdots\left(\Sl_{\sigma_{m-1}}\mathcal{J}_{m-1}\right)\\
  \end{split}
\end{equation}
Then, using the result from the generalized $U(1)$-decoupling identity and recursion hypothesis of BCJ relation with less external lines, we have
\begin{equation}\label{J_k-GU(1)}
  \Sl_{\sigma_{k}}\mathcal{J}_{k}=\Sl_{div\{\alpha,\beta\in \mathcal{J}_k\}}\left(\frac{1}{2F^2}\right)^{\frac{R_k+S_k-1}{2}}\mathcal{J}(A_1^k)\cdots \mathcal{J}(A_{R_k}^k)\mathcal{J}(B_1^k)\cdots \mathcal{J}(B_{S_k}^k),
\end{equation}
\begin{equation}\label{J_k-BCJ}
  \Sl_{\sigma_{k}}S_k\mathcal{J}_k=\Sl_{div\{\alpha,\beta\in \mathcal{J}_k\}}\left(\frac{1}{2F^2}\right)^{\frac{R_k+S_k-1}{2}}S_{div\{\alpha,\beta\in \mathcal{J}_k\}}^\mathrm{II}\mathcal{J}(A_1^k)\cdots \mathcal{J}(A_{R_k}^k)\mathcal{J}(B_1^k)\cdots \mathcal{J}(B_{S_k}^k),
\end{equation}
where $\{A_1^k,\cdots,A_{R_k}^k,B_1^k,\cdots,B_{S_k}^k\}$ is a division of the external lines in $\mathcal{J}_{k}$.
Substitute (\ref{J_k-GU(1)}) and (\ref{J_k-BCJ}) into (\ref{$A_{RBCJ2}$-splitted}), and we arrive at
\begin{eqnarray}\label{arrive}
 &~& \mathcal{A}_{RBCJ}^\mathrm{II}(r,s)=
  \Sl_{m=4}^n\Sl_{div}\Sl_{\bar{\sigma}}V(m,div,\bar{\sigma})\left(\Sl_{div\{\alpha,\beta\in \mathcal{J}_{m-1}\}}\cdots\Sl_{div\{\alpha,\beta\in \mathcal{J}_{1}\}}\right.\nb\\
 &~& \left.(\Sl_{k=1}^{m-1}S_{div\{\alpha,\beta\in \mathcal{J}_k\}}^\mathrm{II}+\bar{S})\prod_{l=1}^{m-1}\left(\frac{1}{2F^2}\right)^{\frac{R_l+S_l-1}{2}}\mathcal{J}(A_1^l)\cdots \mathcal{J}(A_{R_l}^l)\mathcal{J}(B_1^l)\cdots \mathcal{J}(B_{S_l}^l)\right)\nb\\
\end{eqnarray}
The formula above contains all the possible division. If we concentrate on a specific division of sub-currents $\prod\limits_{l=1}^{m-1}\left(\frac{1}{2F^2}\right)^{\frac{R_l+S_l-1}{2}}\mathcal{J}(A_1^l)\cdots \mathcal{J}(A_{R_l}^l)\mathcal{J}(B_1^l)\cdots \mathcal{J}(B_{S_l}^l)$, the corresponding coefficient
\begin{equation}\label{sum_coef}
  \mathcal{VS}_{div}(r,s)\equiv\Sl_{m=4}^n\Sl_{div}\Sl_{\bar{\sigma}}V(m,div,\bar{\sigma})\left(\Sl_{k=1}^{m-1}S_{div\{\alpha,\beta\in \mathcal{J}_k\}}^\mathrm{II}+\bar{S}\right)
\end{equation}
is quartic polynomial of external momenta.

In a nutshell, according to Berends-Giele recursion relation, generalized $U(1)$-decoupling identity and recursion hypothesis of the revised BCJ relation, there are no sub-currents with both $\alpha$ and $\beta$ in the summation, and the coefficients before we sum up $\Sl_{m=4}^n\Sl_{div}\Sl_{\bar{\sigma}}$ are all quartic polynomials of external momenta.
\end{proof}

Now let us concentrate on the summation $\Sl_{m=4}^n\Sl_{div}\Sl_{\bar{\sigma}}$. For convenience, we can denote the specific division $div\left\{\bigcup\limits_{l=1}^{m-1}A_1^l\cdots A_{R_l}^lB_1^l\cdots B_{S_l}^l\right\}$
by $div\{A_{1}\cdots A_RB_{1}\cdots B_S\}$, where $\Sl_{l=1}^{m-1}R_l=R$, $\Sl_{l=1}^{m-1}S_l=S$. The summation may bring about $P_{A_I}^2(I=1,\cdots,R)$ or $P_{B_J}^2(J=1,\cdots,S)$, then kill the propagator, for example $\frac{1}{P_{A_1}^2}$ in $\mathcal{J}(A_1)$, and make it finer sub-currents $\mathcal{J}(A_{1_1})\cdots \mathcal{J}(A_{1_{R_1}})$ which belong to another division. However, the coefficients in the terms which are belong to another division is still quartic polynomial of external momenta.

\subsection{Completeness and non-repetitiveness of the argument}\label{subsec:Complet}
That the coefficients in each division are quartic polynomials is out of question now. However, one may suspect that the completeness and non-repetitiveness of each division after the summation cannot be guaranteed automatically. In order to make it more manifest, we can equip all the divisions with a partial order structure defined as following,
\begin{definition}\label{partial order}~

\begin{enumerate}
  \item $A_{1}\cdots A_R$ and $A'_{1}\cdots A'_{R'}$ are two divisions of $\alpha_1\cdots\alpha_r$. We say $A'_{1}\cdots A'_{R'}<A_{1}\cdots A_R$, if $\forall~A'_{I'}\in\{A'_1,\cdots,A'_{R'}\}$, $\exists~A_I\in\{A_1,\cdots,A_{R}\}$, s.t. $A'_{I'}\subset A_I$.
  \item $B_{1}\cdots B_S$ and $B'_{1}\cdots B'_{S'}$ are two divisions of $\beta_1\cdots\beta_s$. We say $B'_{1}\cdots B'_{S'}<B_{1}\cdots B_S$, if $\forall~B'_{J'}\in\{B'_1,\cdots,B'_{S'}\}$, $\exists~B_J\in\{B_1,\cdots,B_{S}\}$, s.t. $B'_{J'}\subset B_J$.
  \item $A_{1}\cdots A_RB_{1}\cdots B_S$ and $A'_{1}\cdots A'_{R'}B'_{1}\cdots B'_{S'}$ are two divisions of $\alpha_1\cdots\alpha_r\beta_1\cdots\beta_s$. We say $A'_{1}\cdots A'_{R'}B'_{1}\cdots B'_{S'}<A_{1}\cdots A_RB_{1}\cdots B_S$, if $A'_{1}\cdots A'_{R'}<A_{1}\cdots A_R$ and $B'_{1}\cdots B'_{S'}<B_{1}\cdots B_S$.
\end{enumerate}
\end{definition}
By this definition, only the divisions which are bigger than a certain $div$ have the possibility to kill some propagators and contribute to it. Take $\mathcal{J}(\alpha_1\cdots\alpha_7\beta_1\beta_2)$ for example. The order of $div\{\beta_1\beta_2\}$ is trivial, so we only need to pay attention to the order of $div\{\alpha_1\cdots\alpha_7\}$. Let's introduce a diagram representation of a certain division for convenience. For example, we use Fig.\ref{rep_div} to represent the division $\mathcal{J}(\alpha_1)\mathcal{J}(\alpha_2)\mathcal{J}(\alpha_3)\mathcal{J}(\alpha_4)\mathcal{J}(\alpha_5\alpha_6\alpha_7)$ in our representation. We can draw a Hasse diagram in Fig.\ref{Hasse_7} to express its partial order structure.
\begin{figure}[H]
  \centering
  \includegraphics[width=0.1\textwidth]{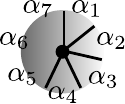}\\
  \caption{Representation of division $\mathcal{J}(\alpha_1)\mathcal{J}(\alpha_2)\mathcal{J}(\alpha_3)\mathcal{J}(\alpha_4)\mathcal{J}(\alpha_5\alpha_6\alpha_7)$ as an example.}\label{rep_div}
\end{figure}

\begin{figure}[htbp]
  \centering
  \includegraphics[width=0.5\textwidth]{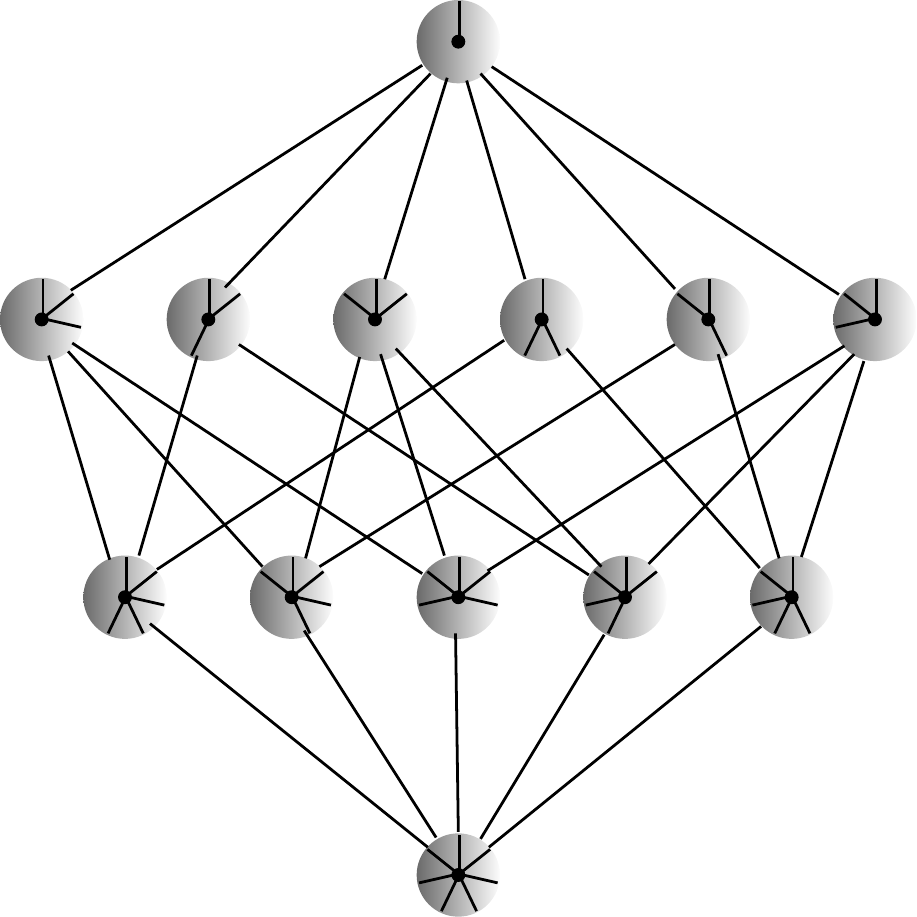}\\
  \caption{The Hasse diagram for the different divisions of $\mathcal{J}(\alpha_1\cdots\alpha_7\beta_1\beta_2)$. The divisions in the same row are of same number of parts. We shall use $\mathbf{i}_j$ to denote the $j$-th division in the $i$-th row for simplicity.}\label{Hasse_7}
\end{figure}

Now let us consider the 5-part division $\mathbf{3}_1$ as an example. The divisions which are bigger than it are 3-part divisions $\mathbf{2}_1$, $\mathbf{2}_2$, $\mathbf{2}_4$ and 1-part division $\mathbf{1}_1$.

First, let us consider the contribution from $\mathbf{2}_1$ to $\mathbf{3}_1$. The diagrams in $\mathbf{2}_1$ which can contribute to $\mathbf{3}_1$ are of the form in the left diagram in Fig.\ref{F2-3}, where $\alpha_3,\cdots,\alpha_7$ are always in one sub-current which is attached to the vertices containing the off-shell line 1. Since we can calculate the term without sub-currents in $\mathcal{VS}_{0}(3,2)$\footnote{Here 0 is used to denote the division without sub-currents, and we will keep this convention in the rest of this paper.} with $\alpha_3$ off-shell,
\begin{figure}[h!]
  \centering
  \includegraphics[width=0.7\textwidth]{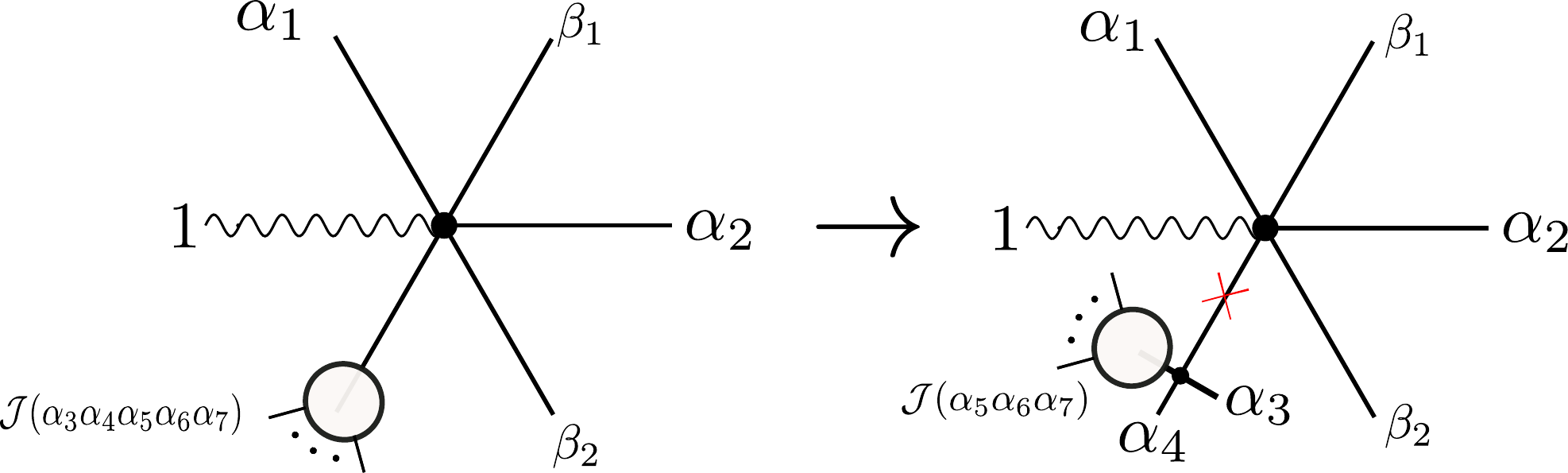}\\
  \caption{Direct contribution to  $\mathbf{3}_1$ from $\mathbf{2}_1$, excluding the contribution from $\mathbf{1}_1$, denoted by $\mathbf{2}_1\rightarrow\mathbf{3}_1$}\label{F2-3}
\end{figure}
\begin{equation}
  \begin{split}
  \mathcal{VS}_{0}(3,2)=-(s_{\alpha_1,\beta_1\beta_2}+s_{\alpha_2,\beta_2})p_1^2+s_{\alpha_1\alpha_2\alpha_3,\beta_1}s_{\alpha_1\alpha_3}+(s_{\alpha_1\alpha_2\alpha_3,\beta_1}+s_{\alpha_1\beta_2})p_{\alpha_3}^2,
  \end{split}
\end{equation}
substitute the off-shell $p_{\alpha_3}$ with $(p_{\alpha_3}+p_{\alpha_4}+p_{\alpha_5}+p_{\alpha_6}+p_{\alpha_7})$, and we can get
\begin{equation}\label{2_1}
  \begin{split}
  &\mathcal{VS}_{\mathbf{2}_1}(7,2)\mathcal{J}(\alpha_1)\mathcal{J}(\alpha_2)\mathcal{J}(\alpha_3\cdots\alpha_7)\mathcal{J}(\beta_1)\mathcal{J}(\beta_2)\\ =&-\left((s_{\alpha_1,\beta_1\beta_2}+s_{\alpha_2,\beta_2})p_1^2+s_{\alpha_1\cdots\alpha_7,\beta_1}s_{\alpha_1,\alpha_3\cdots\alpha_7}+(s_{\alpha_1\cdots\alpha_7,\beta_1}+s_{\alpha_1\beta_2}) \left(p_{\alpha_3}+p_{\alpha_4}+p_{\alpha_5}+p_{\alpha_6}+p_{\alpha_7}\right)^2\right)\\
  &\times\mathcal{J}(\alpha_1)\mathcal{J}(\alpha_2)\mathcal{J}(\alpha_3\cdots\alpha_7)\mathcal{J}(\beta_1)\mathcal{J}(\beta_2)
  \end{split}
\end{equation}
The first coefficient in (\ref{2_1}) remains in the division $\mathbf{2}_1$, the second one is killed by the contribution from $\mathbf{1}_1$ to $\mathbf{2}_1$ as is calculated in (\ref{1-2}), while the third one contributes from $\mathbf{2}_1$ to $\mathbf{3}_1$, $\mathbf{3}_2$ and $\mathbf{3}_3$. Therefore we have,
\begin{equation}\label{2_1'}
  \begin{split}
  \mathcal{A}_{RBCJ\mathbf{2}_1}^\mathrm{II}(7,2)= -(s_{\alpha_1,\beta_1\beta_2}+s_{\alpha_2,\beta_2})p_1^2\mathcal{J}(\alpha_1)\mathcal{J}(\alpha_2)\mathcal{J}(\alpha_3\cdots\alpha_7)\mathcal{J}(\beta_1)\mathcal{J}(\beta_2)
  \end{split}
\end{equation}
The contribution from $\mathbf{2}_1$ to $\mathbf{3}_1$ (Fig.\ref{F2-3}) is
\begin{equation}\label{2-3}
  \mathbf{2}_1\rightarrow\mathbf{3}_1: (s_{\alpha_1\cdots\alpha_7,\beta_1}+s_{\alpha_1\beta_2})\left(p_{\alpha_3}+p_{\alpha_5}+p_{\alpha_6}+p_{\alpha_7}\right)^2.
\end{equation}

Similarly, we can get the contribution from $\mathbf{2}_2$ (Fig.\ref{F2'-3}) and $\mathbf{2}_4$ (Fig.\ref{F2''-3}) to $\mathbf{3}_1$,
\begin{figure}[h!]
  \centering
  \includegraphics[width=0.7\textwidth]{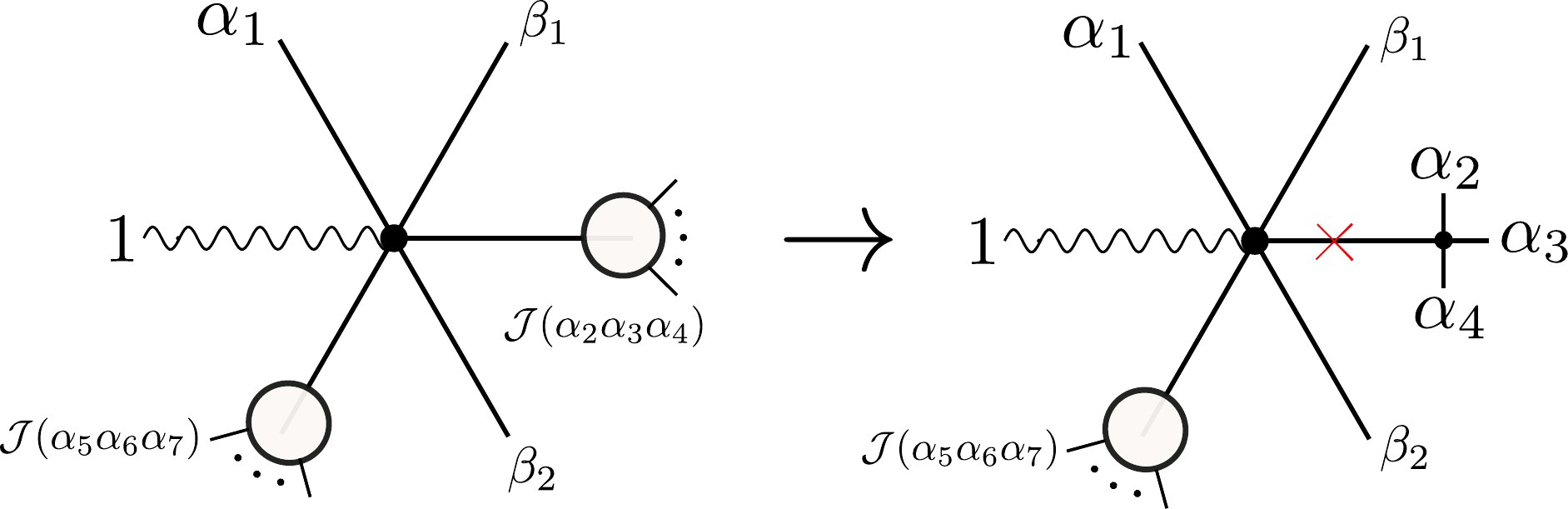}\\
  \caption{$\mathbf{2}_2\rightarrow\mathbf{3}_1$}\label{F2'-3}
\end{figure}
\begin{figure}[h!]
  \centering
  \includegraphics[width=0.7\textwidth]{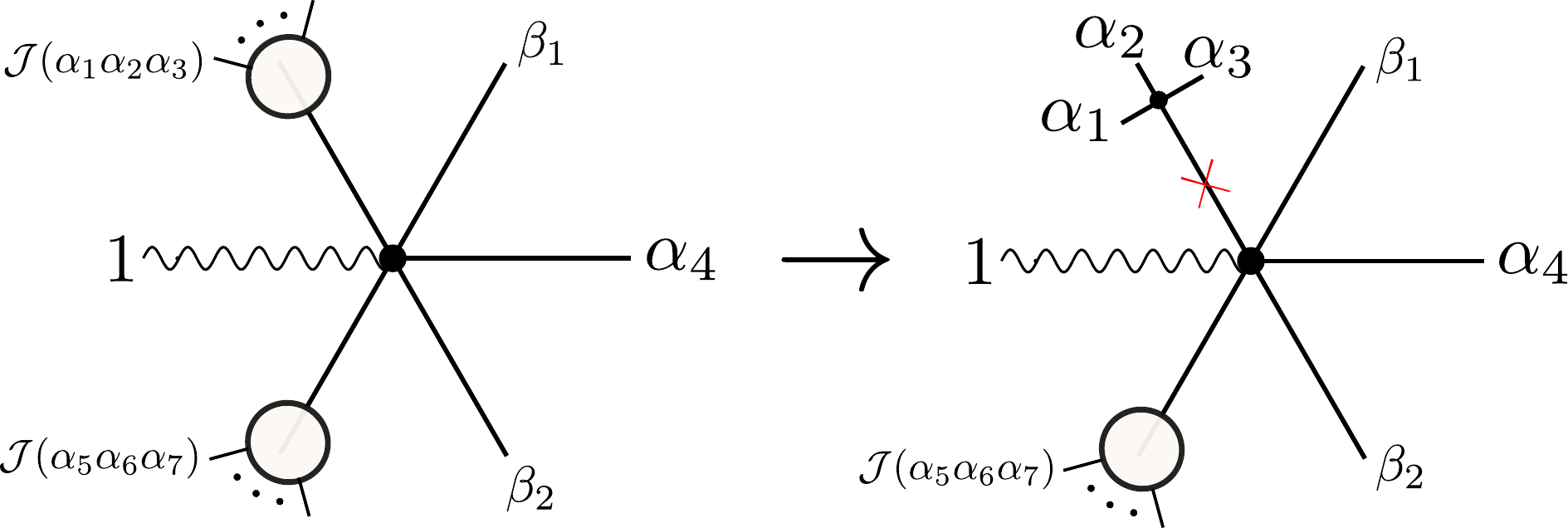}\\
  \caption{$\mathbf{2}_4\rightarrow\mathbf{3}_1$}\label{F2''-3}
\end{figure}
\begin{equation}\label{2-3'}
  \begin{split}
  &\mathbf{2}_2\rightarrow\mathbf{3}_1: (s_{\alpha_1,\beta_2}+2s_{\alpha_2\alpha_3\alpha_4,\beta_2})\left(p_{\alpha_2}+p_{\alpha_4}\right)^2\\
  &\mathbf{2}_4\rightarrow\mathbf{3}_1: (s_{\alpha_1\cdots\alpha_7,\beta_1}+s_{\alpha_1\alpha_2\alpha_3\alpha_4,\beta_1}+s_{\alpha_1\cdots\alpha_7,\beta_2})\left(p_{\alpha_1}+p_{\alpha_3}\right)^2,\\
  \end{split}
\end{equation}

Second, let us consider the contribution from $\mathbf{1}_1$ to $\mathbf{3}_1$. They are composed of four parts, contribution from $\mathbf{1}_1$ to $\mathbf{3}_1$ via $\mathbf{2}_1$ (Fig.\ref{F1-2-3}), contribution from $\mathbf{1}_1$ to $\mathbf{3}_1$ via $\mathbf{2}_2$ (Fig.\ref{F1-2'-3}), contribution from $\mathbf{1}_1$ to $\mathbf{3}_1$ via $\mathbf{2}_4$ (Fig.\ref{F1-2''-3}) and contribution from $\mathbf{1}_1$ to $\mathbf{3}_1$ directly (Fig.\ref{F1-3}).Again, $\mathcal{VS}_{0}(1,2)$ with $\alpha_1$ off-shell can be calculated directly,
\begin{equation}
  \mathcal{VS}_{0}(1,2)\mathcal{J}(\alpha_1)\mathcal{J}(\beta_1)\mathcal{J}(\beta_2) =-(s_{\alpha_1\beta_2}p_1^2+s_{\alpha_1\beta_1}p_{\alpha_1}^2)\mathcal{J}(\alpha_1)\mathcal{J}(\beta_1)\mathcal{J}(\beta_2).
\end{equation}
Substitute the off-shell $p_{\alpha_1}$ with $(p_{\alpha_1}+p_{\alpha_2}+p_{\alpha_3}+p_{\alpha_4}+p_{\alpha_5}+p_{\alpha_6}+p_{\alpha_7})$, and we can get the contribution from $\mathbf{1}_1$ to $\mathbf{2}_1$, $\mathbf{2}_2$ and $\mathbf{2}_4$,
\begin{equation}\label{}
  \begin{split}
  &\mathbf{1}_1\rightarrow\mathbf{2}_1: -s_{\alpha_1\cdots\alpha_7,\beta_1}\left(p_{\alpha_1}+p_{\alpha_3}+p_{\alpha_4}+p_{\alpha_5}+p_{\alpha_6}+p_{\alpha_7}\right)^2\\
  &\mathbf{1}_1\rightarrow\mathbf{2}_2: -s_{\alpha_1\cdots\alpha_7,\beta_1}\left(p_{\alpha_1}+p_{\alpha_5}+p_{\alpha_6}+p_{\alpha_7}\right)^2\\
  &\mathbf{1}_1\rightarrow\mathbf{2}_4: -s_{\alpha_1\cdots\alpha_7,\beta_1}\left(p_{\alpha_1}+p_{\alpha_2}+p_{\alpha_3}+p_{\alpha_5}+p_{\alpha_6}+p_{\alpha_7}\right)^2\\
  \end{split}
\end{equation}
Take $\mathbf{1}_1\rightarrow\mathbf{2}_1$ as an example,
\begin{equation}\label{1-2}
  -s_{\alpha_1\cdots\alpha_7,\beta_1}s_{\alpha_1,\alpha_3\cdots\alpha_7}
\end{equation}
will remain in $\mathbf{2}_1$, and since it does not contain cross terms of $\alpha_3,\cdots,\alpha_7$, it will never contribute to the terms which can offset the propagator in $\mathcal{J}(\alpha_3\cdots\alpha_7)$.
\begin{equation}\label{}
  -s_{\alpha_1\cdots\alpha_7,\beta_1}\left(p_{\alpha_3}+p_{\alpha_4}+p_{\alpha_5}+p_{\alpha_6}+p_{\alpha_7}\right)^2
\end{equation}
will offset the propagator in $\mathcal{J}(\alpha_3\cdots\alpha_7)$, and part of it will contribute to $\mathbf{3}_1$ (Fig.\ref{F1-2-3}),
\begin{equation}\label{1-2-3}
  \mathbf{1}_1\xrightarrow{\mathbf{2}_1}\mathbf{3}_1: -s_{\alpha_1\cdots\alpha_7,\beta_1}\left(p_{\alpha_3}+p_{\alpha_5}+p_{\alpha_6}+p_{\alpha_7}\right)^2
\end{equation}
\begin{figure}[h!]
  \centering
  \includegraphics[width=0.8\textwidth]{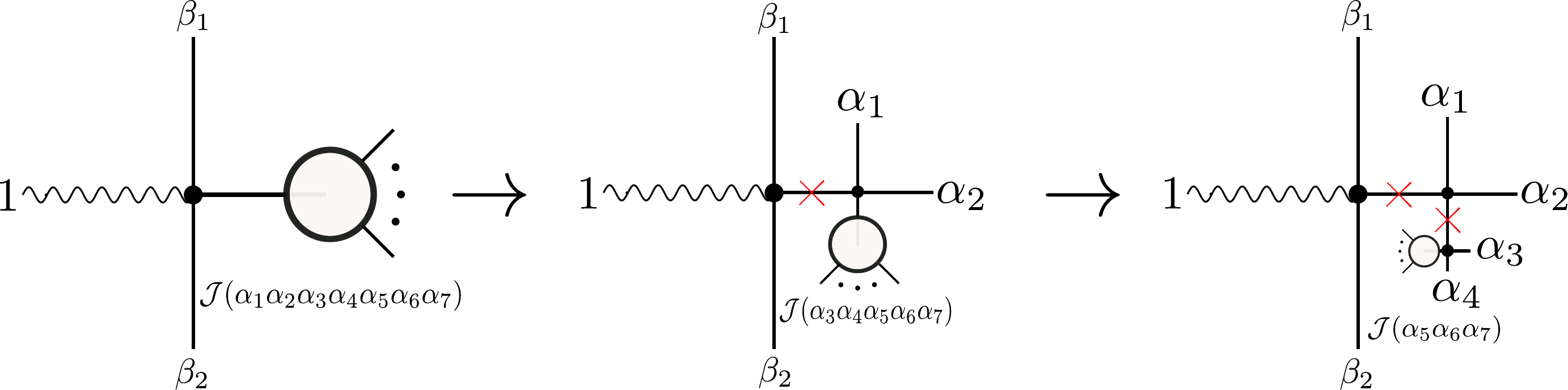}\\
  \caption{The contribution to $\mathbf{3}_1$ from $\mathbf{1}_1$ via $\mathbf{2}_1$, denoted by $\mathbf{1}_1\xrightarrow{\mathbf{2}_1}\mathbf{3}_1$.}\label{F1-2-3}
\end{figure}
Similarly, the contribution from $\mathbf{1}_1$ to $\mathbf{3}_1$ via $\mathbf{2}_2$ (Fig.\ref{F1-2'-3}) and $\mathbf{2}_4$ (Fig.\ref{F1-2''-3}) is
\begin{equation}\label{1-2'-3}
  \begin{split}
  &\mathbf{1}_1\xrightarrow{\mathbf{2}_2}\mathbf{3}_1: 0\\
  &\mathbf{1}_1\xrightarrow{\mathbf{2}_4}\mathbf{3}_1: -s_{\alpha_1\cdots\alpha_7,\beta_1}\left(p_{\alpha_1}+p_{\alpha_3}\right)^2
  \end{split}
\end{equation}
\begin{figure}[h!]
  \centering
  \includegraphics[width=0.6\textwidth]{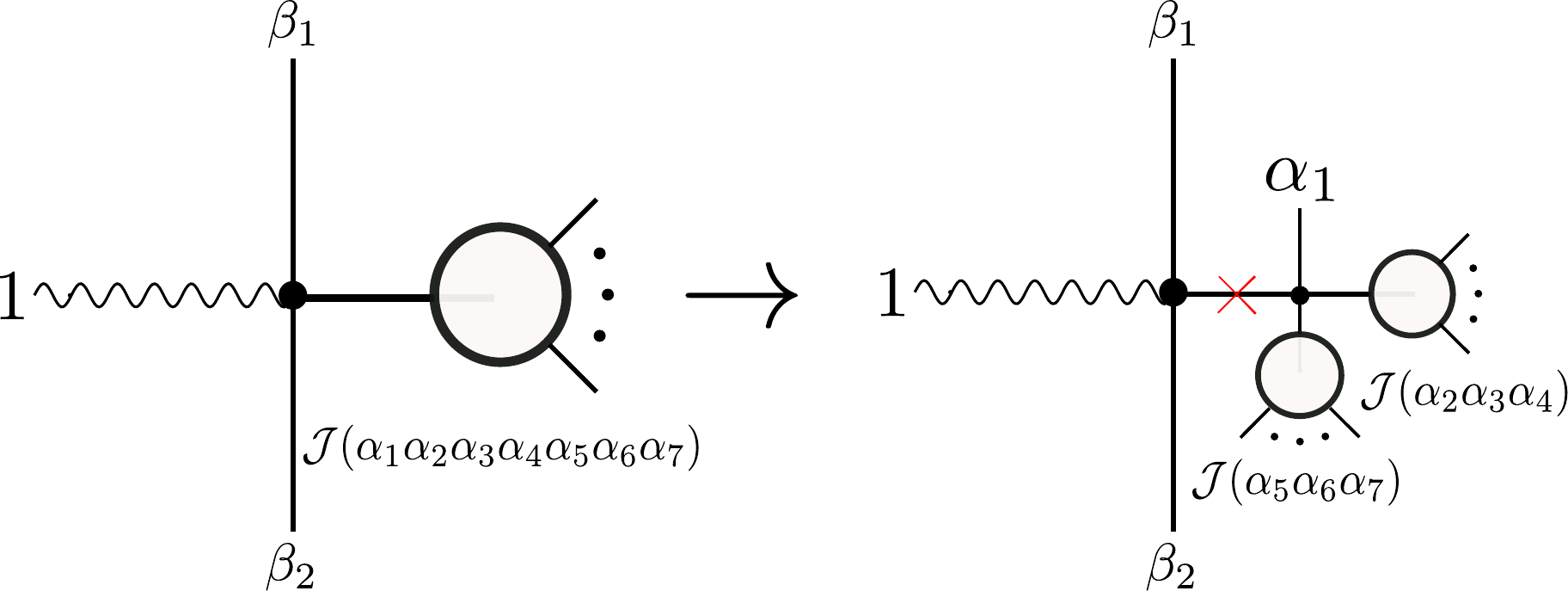}\\
  \caption{$\mathbf{1}_1\xrightarrow{\mathbf{2}_2}\mathbf{3}_1$}\label{F1-2'-3}
\end{figure}
\begin{figure}[h!]
  \centering
  \includegraphics[width=0.8\textwidth]{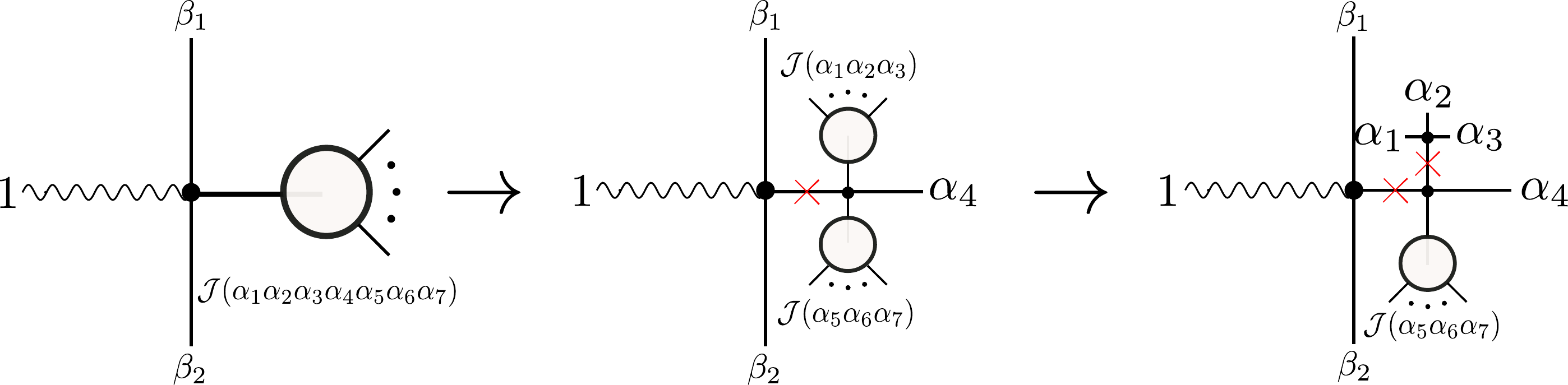}\\
  \caption{$\mathbf{1}_1\xrightarrow{\mathbf{2}_4}\mathbf{3}_1$}\label{F1-2''-3}
\end{figure}
As is illustrated in Fig.\ref{F1-3}, the direct contribution from $\mathbf{1}_1$ to $\mathbf{3}_1$ is,
\begin{equation}\label{1-3}
  \mathbf{1}_1\rightarrow\mathbf{3}_1=-s_{\alpha_1\cdots\alpha_7,\beta_1}\left(p_{\alpha_1}+p_{\alpha_3}+p_{\alpha_5}+p_{\alpha_6}+p_{\alpha_7}\right)^2
\end{equation}
\begin{figure}[h!]
  \centering
  \includegraphics[width=0.7\textwidth]{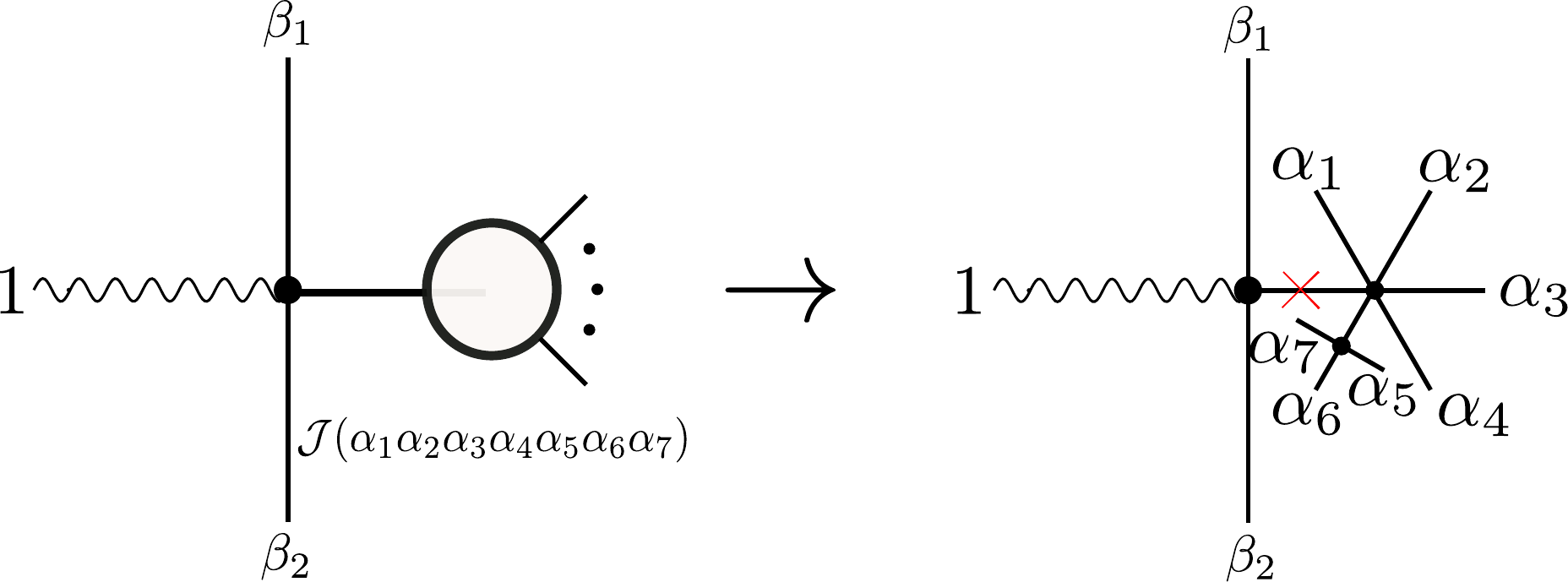}\\
  \caption{Direct contribution to $\mathbf{3}_1$ from $\mathbf{1}_1$, denoted by $\mathbf{1}_1\rightarrow\mathbf{3}_1$}\label{F1-3}
\end{figure}

Finally, combine (\ref{2-3}), (\ref{2-3'}), (\ref{1-2-3}), (\ref{1-2'-3}) and (\ref{1-3}), and we get the total contribution to $\mathbf{3}_1$ is quartic polynomial of external momenta,
\begin{equation}\label{}
  \begin{split}
  &(\mathbf{2}_{1,2,3}\rightarrow\mathbf{3}_1)+(\mathbf{1}_{1}\xrightarrow{\mathbf{2}_{1,2,3}}\mathbf{3}_1)+(\mathbf{1}_1\rightarrow\mathbf{3}_1)\\
  =&s_{\alpha_1\beta_2}\left(p_{\alpha_3}+p_{\alpha_5}+p_{\alpha_6}+p_{\alpha_7}\right)^2+(s_{\alpha_1,\beta_2}+2s_{\alpha_2\alpha_3\alpha_4,\beta_2})\left(p_{\alpha_2}+p_{\alpha_4}\right)^2\\
  &+(s_{\alpha_1\alpha_2\alpha_3\alpha_4,\beta_1}+s_{\alpha_1\cdots\alpha_7,\beta_2})\left(p_{\alpha_1}+p_{\alpha_3}\right)^2-s_{\alpha_1\cdots\alpha_7,\beta_1}\left(p_{\alpha_1}+p_{\alpha_3}+p_{\alpha_5}+p_{\alpha_6}+p_{\alpha_7}\right)^2\\
  \end{split}
\end{equation}

From this example, we found that if $\mathbf{i}_{j}>\mathbf{i'}_{j'}$, $\mathbf{i}_{j}$ would have several paths to contribute to $\mathbf{i'}_{j'}$. However, $s_{\alpha\beta}$ in (\ref{sum_coef}) cannot be used to kill propagators, so it remains in every step. Only the vertex has the possibility to kill the propagator in the sub-current, and thus there cannot be more than one propagator being killed at one vertex. Failure to kill one propagator in one step results in the failure to kill the propagator in the following step. As a result, only if $\mathbf{i'}_{j'}$ is a refined division of $\mathbf{i}_{j}$ with respect to only one sub-current in $\mathbf{i}_{j}$, $\mathbf{i}_{j}$ has the possibility to contribute to $\mathbf{i'}_{j'}$. In our example, $\mathbf{2}_{2}$, $\mathbf{2}_{4}$, $\mathbf{2}_{5}$ will not contribute to $\mathbf{4}_{1}$. This result motivates us to define a sub-order $\prec$ of the partial order $<$ defined in Definition \ref{partial order},
\begin{definition}~

\begin{enumerate}
  \item $A_{1}\cdots A_R$ and $A'_{1}\cdots A'_{R'}$ are two divisions of $\alpha_1\cdots\alpha_r$. We say $A'_{1}\cdots A'_{R'}\prec A_{1}\cdots A_R$, if for all $A_{I}\in\{A_1,\cdots,A_{R}\}$ except one, say $A_{I_1}$, $\exists~A'_{I'}\in\{A'_1,\cdots,A_{R'}\}$, s.t. $A_I=A'_{I'}$.
  \item $B_{1}\cdots B_S$ and $B'_{1}\cdots B'_{S'}$ are two divisions of $\beta_1\cdots\beta_s$. We say $B'_{1}\cdots B'_{S'}\prec B_{1}\cdots B_S$, if for all $B_{J}\in\{B_1,\cdots,B_{S}\}$ except one, say $B_{J_1}$, $\exists~B'_{J'}\in\{B'_1,\cdots,B_{S'}\}$, s.t. $B_J=B'_{J'}$.
  \item $A_{1}\cdots A_RB_{1}\cdots B_S$ and $A'_{1}\cdots A'_{R'}B'_{1}\cdots B'_{S'}$ are two divisions of $\alpha_1\cdots\alpha_r\beta_1\cdots\beta_s$. We say $A'_{1}\cdots A'_{R'}B'_{1}\cdots B'_{S'}\prec A_{1}\cdots A_RB_{1}\cdots B_S$, if one of the following conditions are satisfied,
      \begin{itemize}
        \item $A'_{1}\cdots A'_{R'}=A_{1}\cdots A_R$ and $B'_{1}\cdots B'_{S'}\prec B_{1}\cdots B_S$.
        \item $B'_{1}\cdots B'_{S'}=B_{1}\cdots B_S$ and $A'_{1}\cdots A'_{R'}\prec A_{1}\cdots A_R$
      \end{itemize}
\end{enumerate}
\end{definition}
In another words, only if $\mathbf{i}_{j}\succ\mathbf{i'}_{j'}$ holds, $\mathbf{i}_{j}$ has the possibility to contribute to $\mathbf{i'}_{j'}$. $\prec$ is not a partial order, because it is not transitive, but we only consider the paths where its transitivity holds. If $\mathbf{i}_{j}=A_1\cdots A_{L-1}A_LA_{L+1}\cdots B_S$, the general form of $\mathbf{i'}_{j'}$ to which $\mathbf{i}_{j}$ can contribute is illustrated in Fig. \ref{general}. That is $\mathbf{i'}_{j'}=A_1\cdots A_{L-1}\cdots A_{L1}\cdots A_{LM-1}A_{LM}\cdots A_{LM1}\cdots A_{LMN}\cdots A_{L+1}\cdots B_S$.
\begin{figure}[h!]
  \centering
  \includegraphics[width=0.8\textwidth]{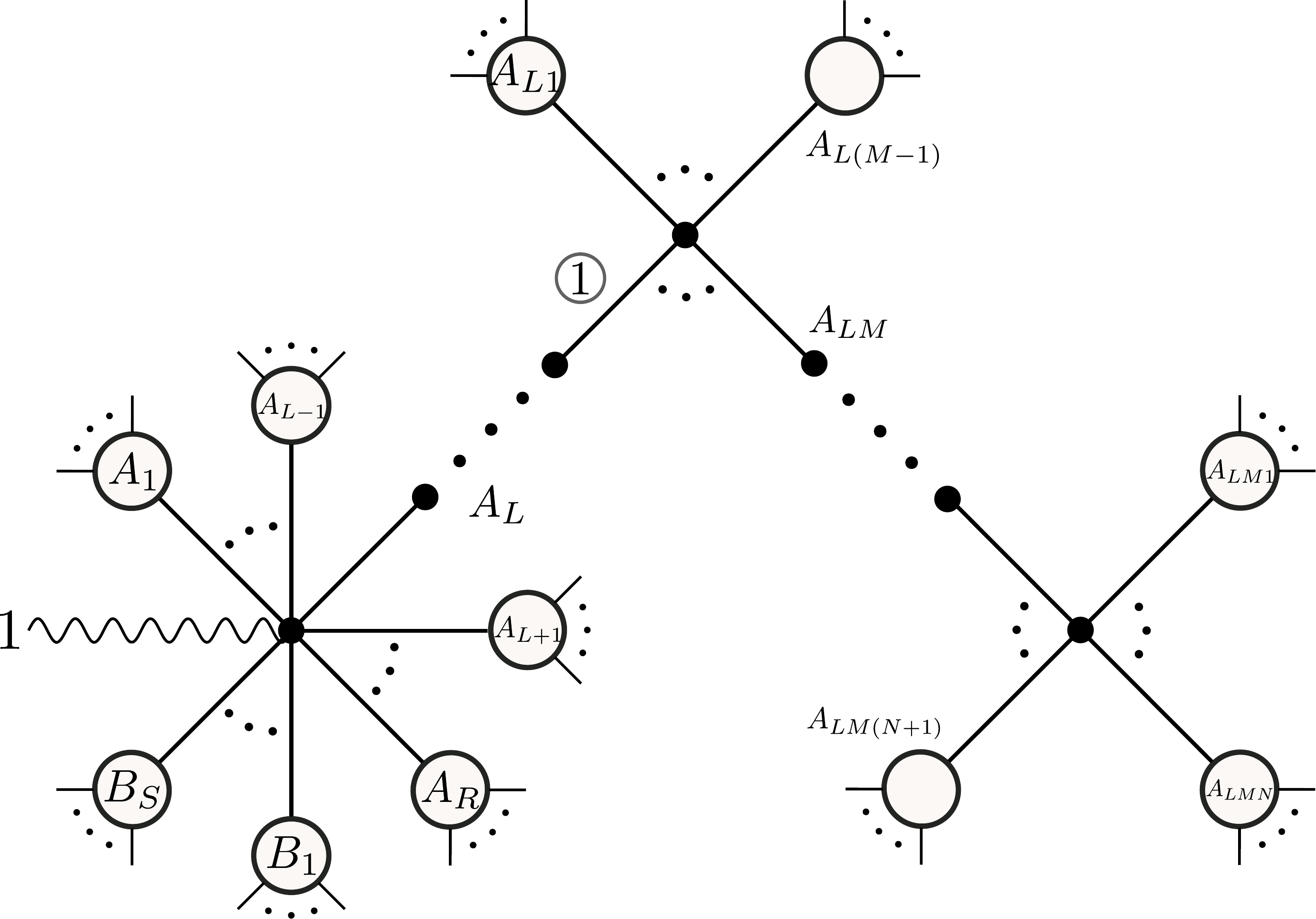}\\
  \caption{The diagram corresponding to the contribution from division $\mathbf{i}_{j}$ to division $\mathbf{i'}_{j'}$}\label{general}
\end{figure}

For each totally ordered path under $\prec$ in the Hasse diagram, there is a counterpart in the diagrams like Fig.\ref{general}. Let us see why it is true. After the propagator \ding{172} is killed, the coefficients are the product of $s_{\alpha\beta}$ and the vertex. If the momentum $P_{A_{LM}}$ is in the vertex, it will provide a $P_{A_{LM}}^2$ to kill the propagator $\frac{1}{P_{A_{LM}}^2}$ and go ahead to next step. Otherwise $\frac{1}{P_{A_{LM}}^2}$ will not be killed, and thus there is no contribution to $ \mathbf{i'}_{j'}$.   Whether $\frac{1}{P_{A_{LM}}^2}$ has been killed or not, the remainder in the vertex does not contain cross terms of $p$ in $A_{LM}$. So the remainder will not provide terms to kill $\frac{1}{P_{A_{LM}}^2}$ any more, which means all the terms which can kill the propagator have been considered. Now we know that the counterpart is just the path to kill the propagators successively.

Different paths correspond to different ways to kill the propagators, and thus different contributions to $\mathbf{i'}_{j'}$. Since there is a one-to-one correspondence between the totally ordered paths under $\prec$ and the paths of propagators to be killed, we can attribute each contribution to its corresponding diagram. Because all the diagrams are complete and not repetitive, we claim that all the contributions are complete and not repetitive.
%%$\mathcal{VS}_{\mathbf{2}_{j_2}}$ is quartic polynomial of external momenta by Lemma \ref{sum_coef}. Thus after $\mathcal{VS}_{\mathbf{2}_{j_2}}$ have get the contribution from division bigger than it\footnote{Denoted by $\mathcal{VS}_{\mathbf{2}_{j_2}'}$}, it is still quartic polynomial of external momenta. Then suppose that $\mathcal{VS}_{\mathbf{i}_{j_i}'}(i=1,\cdots,n-1)$ are quartic polynomials of external momenta, and thus their contributions to $\mathcal{VS}_{\mathbf{n}_{j_n}}$ are quartic polynomials of external momenta. Since $\mathcal{VS}_{\mathbf{n}_{j_n}}$ is quartic polynomial of external momenta by Lemma \ref{sum_coef}, we know that after $\mathcal{VS}_{\mathbf{n}_{j_n}}$ have get the contribution from division bigger than it, $\mathcal{VS}_{\mathbf{n}_{j_n}'}$ is quartic polynomial of external momenta. Therefore we conclude that all the $\mathcal{VS}_{\mathbf{i}_{j_i}'}(i=1,\cdots,r+s-2)$ are quartic polynomials of external momenta.

Let us make a summary about the coefficients of each division in the revised BCJ relation. The coefficients of a specific division $\mathbf{i'}_{j'}$ can get contributions from divisions which are bigger than it under the order $\succ$. The contributions are quartic polynomials of external momenta, complete and not repetitive. Summing them up with the original part in $\mathbf{i'}_{j'}$ which cannot contribute to smaller divisions, we will obtain the final coefficients of the division $\mathbf{i'}_{j'}$, which are quartic polynomials of external momenta, complete and not repetitive.

\subsection{Sub-current recursion and the general form of the coefficients of the divisions without sub-currents}\label{subsec:CurrentRec}
A direct result of this argument is the following lemma,
\begin{lemma}\label{sub-currents recursion}
The coefficients of the divisions with sub-currents can be obtained recursively from the coefficients of the divisions without sub-currents.
\end{lemma}
\begin{proof}
If we know the coefficients of the divisions without sub-currents is
\begin{equation}\label{VS_0'}
  \mathcal{VS}_{0'}\equiv\mathcal{VS}_{\mathbf{r+s-2}_{j_{r+s-2}}'}=S_0p_1^2,
\end{equation}
then $\mathcal{VS}_{0'}$ with all the external line off-shell must have the form
\begin{equation}
  S_0p_1^2+\Sl_{i=1}^{r}a_i p^2_{\alpha_i}+\sum\limits_{j=1}^{s}b_j p^2_{\beta_j}.
\end{equation}
This is because when we take the limit of all the $\alpha$ and $\beta$ being on-shell, it must be reduced to (\ref{VS_0'}). After we substitute those off-shell lines with sub-currents, the terms proportional to $p^2_{A_i}$ will offset the propagators of sub-current $\mathcal{J}(\{A_i\})$. Then such term will contribute to finer divisions. Finally we can get the part with sub-currents,
\begin{eqnarray}\label{}
  \mathcal{A}_{RBCJ}(r,s)=&\Sl_{{divisions\{\alpha,\beta\}\atop{{\abs{R-S}=1}\atop{R+S<r+s}}}}p_{1}^2S_{div}\mathcal{J}(\{A_{1}\})\cdots \mathcal{J}(\{A_{R}\}) \mathcal{J}(\{B_{1}\})\cdots \mathcal{J}(\{B_{S}\}) \nb\\
  &+\mathcal{VS}_{0'}\mathcal{J}(\alpha_1)\cdots \mathcal{J}(\alpha_r) \mathcal{J}(\beta_1)\cdots \mathcal{J}(\beta_s)
\end{eqnarray}
\end{proof}

Thus our next aim is to prove the form of $\mathcal{VS}_{0'}$. We have the following claim to simplify our proof. Since we have known the result of the on-shell general BCJ relation, we can claim

\paragraph{Claim} $\mathcal{VS}_{0'}$ must have at least one factor of $p_1^2$ if it does not vanish.

If it vanishes, it also makes sense to say that it has a factor of $p_1^2$. While if it does not vanish, the claim can be shown easily by \textit{reductio ad absurdum}: If the whole amplitude does not have a factor of $p_1^2$, the relation will contradict to the on-shell case when we take the on-shell limit $p_1^2\rightarrow 0$.

With this claim, we can easily obtain that $\mathcal{VS}_{0'}$ must have the form $p_1^2S_0$, where the general form of $S_0$ is
\begin{equation}\label{}
  S_0=c^{ij}s_{\alpha_i\beta_j}+\bar{c}^{ij}s_{\alpha_i\alpha_j}+\bar{\bar{c}}^{ij}s_{\beta_i\beta_j} ,\ c^{ij}\in\mathbb{Z},\ i=1,\cdots,r\ \mathrm{and}\ j=1,\cdots, s\footnote{The tensor notation $c^{ij}$ refers to an element in the matrix $\{c^{ij}\}_{\substack{i=1,\cdots,r\\j=1,\cdots, s}}$, and boldface $\mathbf{C}(r,s)$ refers to the corresponding matrix. We will keep using this convention in the rest of this paper.}
\end{equation}
The term $s_{p_1*}$ does not appear because we use the convention to read the coefficients backwards(with respect to the general BCJ relation). Yet we have a much stronger constrain,

\begin{lemma}\label{$A_{RBCJ2}$}
The result of $\mathcal{A}_{RBCJ}^\mathrm{II}(r,s)$ only contains coefficients $s_{\alpha_i\beta_j}$, i.e. no $s_{\alpha_i\alpha_j}$ or $s_{\beta_i\beta_j}$.
\end{lemma}
\begin{proof}
We can expand $p_1^2$ in $\mathcal{A}_{RBCJ}(r,s)$, so no $s_{\beta_i\beta_j}$ means that no $s_{\beta_i\beta_j}^2$ in $\mathcal{A}_{RBCJ}(r,s)$. Each diagram in $\mathcal{A}_{RBCJ}(r,s)$ contains coefficients $s_{\alpha_i\beta_j}$ and $s_{\alpha_i\alpha_j}$. The summation of some diagrams will cancel some sub-currents by our recursion hypothesis. In these processes, cancellation of coefficients $s_{\alpha_i*}$ will create a new coefficient $s_{\alpha_i*}$. As a result, $s_{\beta_i\beta_j}^2$ does not appear, so there is no coefficient $s_{\beta_i\beta_j}$ in $\mathcal{A}_{RBCJ}(r,s)$, thus neither in $\mathcal{A}_{RBCJ}^\mathrm{II}(r,s)$.

Similarly, $\mathcal{A}_{RBCJ}(r,\dot{s})$ can be divided into $\mathcal{A}_{RBCJ}^\mathrm{I}(r,\dot{s})=\left(\Sl_{j=1}^s\beta_j\right)^2{\mathcal{A}_{GU}}(r,\dot{s})$ and $\mathcal{A}_{RBCJ}^\mathrm{II}(r,\dot{s})$. If $\mathcal{A}_{RBCJ}^\mathrm{II}(r,s)$ contains terms with coefficients $s_{\alpha_i\alpha_j}$, $\mathcal{A}_{RBCJ}^2(r,\dot{s})$ must contain terms with coefficients $s_{\beta_i\beta_j}$ for symmetry. However, use (\ref{BCJ-GU(1)}) and we can get,
\begin{equation}\label{}
  \mathcal{A}_{RBCJ}^\mathrm{II}(r,s)+\mathcal{A}_{RBCJ}^\mathrm{II}(r,\dot{s})=\Sl_{\sigma\in OP(\mathbf{\alpha_r}\bigcup\mathbf{\beta_s})}\left(\Sl_{i=1}^r\Sl_is_{\alpha_l\beta_i}\right){\mathcal{A}_{GU}}(r,s).
\end{equation}
which means the sum of $\mathcal{A}_{RBCJ}^\mathrm{II}(r,s)$ and $\mathcal{A}_{RBCJ}^\mathrm{II}(r,\dot{s})$ does not contain terms with coefficients $s_{\alpha_i\alpha_j}$ or $s_{\beta_i\beta_j}$. Thus if $\mathcal{A}_{RBCJ}^\mathrm{II}(r,s)$ contains coefficients $s_{\alpha_i\alpha_j}$, it must contain coefficients $-s_{\beta_i\beta_j}$ to cancel the corresponding coefficients caused by $\mathcal{A}_{RBCJ}^\mathrm{II}(r,\dot{s})$, which contradicts with our previous result.
\end{proof}

%Now we only need to pay attention to $\mathcal{A}_{RBCJ}^\mathrm{II}(r,s)$ whose coefficients contain only $s_{\alpha_i\beta_j}$. %If we want to calculate the coefficient of $\mathcal{J}(A_{1})\cdots \mathcal{J}(A_R)\mathcal{J}(B_{1})\cdots \mathcal{J}(B_S)$ which has at least one sub-current, say $\mathcal{J}(A_1)$ without loss of generality, the following lemma makes it possible to use the result of $\mathcal{J}(\alpha_{1})\cdots \mathcal{J}(A_R)\mathcal{J}(B_{1})\cdots \mathcal{J}(B_S)$ to calculate it.

\subsection{The coefficient matrix $\mathbf{C}(r,s)$}\label{subsec:COM}
\begin{prop}
The coefficient matrix $\mathbf{C}(r,s)$ for $S_0(r,s)$ is
\begin{eqnarray}
&&c^{ij}=\left\{
  \begin{array}{ll}
  1,\ if\ i<j \\
  0,\ otherwise
  \end{array}
  \right.,\ if\ r-s=-1, \\
&&c^{ij}=\left\{
  \begin{array}{ll}
  1,\ if\ i\leqslant j \\
  0,\ otherwise
  \end{array}
  \right.,\ if\ r-s=1, \\
&&c^{ij}=0,\ \ \ \ \ \ \ \ \ \ \ \ \ \ \ \ \ \ \ otherwise.
\end{eqnarray}
\end{prop}

\begin{proof}
In order to proof this, we need to introduce another permutation $\tau\in OP(\{\alpha_1\}\bigcup\{\mathbf{\alpha_{r-1}}\})$.
%\begin{gather*}\label{}
  %\mathcal{J}(\alpha_1,\alpha_2,\cdots,\alpha_{r-1},\alpha_{r},\beta_1,\cdots,\beta_s),\\
  %\mathcal{J}(\alpha_2,\alpha_1,\cdots,\alpha_{r-1},\alpha_{r},\beta_1,\cdots,\beta_s),\\
  %\vdots\\
  %\mathcal{J}(\alpha_2,\alpha_3,\cdots,\alpha_{1},\alpha_{r},\beta_1,\cdots,\beta_s),\\
  %\mathcal{J}(\alpha_2,\alpha_3,\cdots,\alpha_{r},\alpha_{1},\beta_1,\cdots,\beta_s).\\
%\end{gather*}
Notice that it is just to rename the $\alpha$'s. Sum them up, and denote summation of $c_{\alpha}^{ij}$ by $d_{\alpha}^{ij}$. Consider $c^{ij}$ with a fixed $i$ and an arbitrary $j$, so there are $(i-1)$ cases where $\alpha_1$ is before $\alpha_i$ and $(r-i+1)$ cases after. Therefore we have
\begin{equation}\label{R}
  \left\{
  \begin{array}{ll}
  d_{\alpha}^{1j}\equiv\ \ \ \ \ \ \ \ \ \ \ \ \sum\limits_{i=1}^rc^{ij}, \ \ \ \ \ \ \ \ \ \ \ \ \ \ \ \ j=1,\cdots,s,\\
  d_{\alpha}^{ij}\equiv(i-1)c^{ij}+(r-i+1)c^{i-1j},\ i=2,\cdots,r\ \mathrm{and}\ j=1,\cdots,s.
  \end{array}
  \right.
\end{equation}
or more explicitly,
\begin{equation*}\label{}
  \boldsymbol{\mathrm{D}_{\alpha}}(r,s)\equiv\mathbf{RC}(r,s),
\end{equation*}
where we denote the transformation matrix by $\mathbf{R}$,
\begin{equation*}
\mathbf{R}=
\begin{pmatrix}
1 & 1 & 1 & \cdots & 1 & 1 & 1\\
r-1 & 1 & 0 & \cdots & 0 & 0 & 0\\
0 & r-2 & 2 & \cdots & 0 & 0 & 0\\
\vdots & \vdots & \vdots & \ddots & \vdots & \vdots & \vdots\\
0 & 0 & 0 & \cdots & 2 & r-2 & 0\\
0 & 0 & 0 & \cdots & 0 & 1 & r-1\\
\end{pmatrix}
\end{equation*}

$\mathbf{R}$ is a matrix with rank $(r-1)$, so we can get the general solution by solving the homogeneous part
\begin{equation}\label{}
  \mathbf{RC}(r,s)=\mathbf{0}
\end{equation}
The solution is
\begin{equation}\label{c_a}
  c^{ij}=n_{\alpha}^j(-1)^i{r-1\choose i-1},\ i=2,\cdots,r.
\end{equation}

The same is for $\beta$, and we can get
\begin{equation}\label{S}
  \left\{
  \begin{array}{ll}
  d_{\beta}^{i1}\equiv \ \ \ \ \ \ \ \ \ \ \  \sum\limits_{j=1}^sc^{ij},\ \ \ \ \ \ \ \ \ \ \ \ \ \ \ \ \ \ \ i=1,\cdots,r,\\
  d_{\beta}^{ij}\equiv(j-1)c^{ij}+(s-j+1)c^{ij-1},\ j=2,\cdots,s\ \mathrm{and}\ i=1,\cdots,r.
  \end{array}
  \right.
\end{equation}
If we denote this transformation matrix by $\mathbf{S}$, then we have
\begin{equation*}\label{}
  \boldsymbol{\mathrm{D}_{\beta}}(r,s)\equiv\mathbf{C}(r,s)\mathbf{S}.
\end{equation*}
And the solution for the homogeneous part is
\begin{equation}\label{c_b}
  c^{ij}=n_{\beta}^i(-1)^j{s-1\choose j-1}
\end{equation}
Combine (\ref{c_a}) and (\ref{c_b}), and we arrive at the general solution
\begin{equation}\label{g}
  c_g^{ij}=n(-1)^{i+j}{r-1\choose i-1}{s-1\choose j-1}
\end{equation}

Now let us deal with the special solution of (\ref{R}) and (\ref{S}). Since the size of $\boldsymbol{\mathrm{D}}(r,s)$ is a function of $(r-s)$, let us denote it by $\boldsymbol{\mathrm{D}^{(r-s)}}(r,s)$.
Even though we do not need to know the expression of $c_{\alpha}^{ij}$, we can calculate $d_{\alpha}^{ij}$ by viewing it in a different way. When an $\alpha$ is permuted among $\mathbf{\alpha_{r-1}}$, we can also think of it as a $\beta$ permuted among $\mathbf{\beta_s}$, since we do not consider the coefficients brought by it. We will calculate the terms without sub-currents from both points of view, compare them, and solve $\boldsymbol{\mathrm{D}^{(r-s)}}(r,s)$. In this part, we assume that $\tau_1\in OP(\{\alpha_1\}\bigcup\{\mathbf{\alpha_{r-1}}\})$, $\tau_2\in OP(\{\alpha_1\}\bigcup\{\mathbf{\beta_s}\})$ without other notations.

From the first point of view,
\begin{equation}\label{b1}
  \begin{split}
  &-\Sl_{\tau_1}\left(\Sl_{\sigma\in OP(\mathbf{\alpha_r}\bigcup\mathbf{\beta_s})}\left(\Sl_{i=1}^r\Sl_{\xi_{\sigma_k}>\xi_{\alpha_i}}s_{\alpha_i\sigma_k}\right)\mathcal{A}(1,\{\sigma\})\right)\\
  =&\Sl_{\tau_1}\left(\Sl_{div\{\mathbf{\alpha_r},\mathbf{\beta_s}\}}   \left(\frac{1}{2F^2}\right)^{\frac{R+S-1}{2}}S_{div\{\mathbf{\alpha_r},\mathbf{\beta_s}\}}\mathcal{J}(A_{1})\cdots \mathcal{J}(A_R)\mathcal{J}(B_{1})\cdots \mathcal{J}(B_S)\right).
  \end{split}
\end{equation}

From the second point of view, take $\alpha_1$ as a member of $\mathbf{\beta_s}$,
\begin{equation}\label{b1}
  \begin{split}
  &-\Sl_{\tau_2}\left(\Sl_{\sigma\in OP(\mathbf{\alpha_{r-1}}\bigcup\mathbf{\beta_{s+1}})}\left(\Sl_{i=2}^r\Sl_{\xi_{\sigma_k}>\xi_{\alpha_i}}s_{\alpha_i\sigma_k}\right)\mathcal{A}(1,\{\sigma\})\right)\\
  =&\Sl_{\tau_2}\left(\Sl_{div\{\mathbf{\alpha_{r-1}},\mathbf{\beta_{s+1}}\}}   \left(\frac{1}{2F^2}\right)^{\frac{R'+S'-1}{2}}S_{div\{\mathbf{\alpha_{r-1}},\mathbf{\beta_{s+1}}\}}\mathcal{J}(A'_{1})\cdots \mathcal{J}(A'_{R'})\mathcal{J}(B'_{1})\cdots \mathcal{J}(B'_{S'})\right).
  \end{split}
\end{equation}
After eliminating the coefficients caused by $\alpha_1$ in the first point of view, which can be evaluated by the fundamental BCJ relation, we arrive at the following equation,
\begin{equation}\label{b1}
  \begin{split}
  &\Sl_{\tau_1}\left(\Sl_{div\{\mathbf{\alpha_r},\mathbf{\beta_s}\}}   \left(\frac{1}{2F^2}\right)^{\frac{R+S-1}{2}}\left(S_{div\{\mathbf{\alpha_r},\mathbf{\beta_s}\}}-S_{\alpha_1\ast}\right)\mathcal{J}(A_{1})\cdots \mathcal{J}(A_R)\mathcal{J}(B_{1})\cdots \mathcal{J}(B_S)\right)\\
  =&\Sl_{\tau_2}\left(\Sl_{div\{\mathbf{\alpha_{r-1}},\mathbf{\beta_{s+1}}\}}   \left(\frac{1}{2F^2}\right)^{\frac{R'+S'-1}{2}}S_{div\{\mathbf{\alpha_{r-1}},\mathbf{\beta_{s+1}}\}}\mathcal{J}(A'_{1})\cdots \mathcal{J}(A'_{R'})\mathcal{J}(B'_{1})\cdots \mathcal{J}(B'_{S'})\right).
  \end{split}
\end{equation}
where $S_{\alpha_1\ast}$ is the coefficient caused by $\alpha_1$ in $S_{div\{\mathbf{\alpha_r},\mathbf{\beta_s}\}}$, and happens to be the same as calculated from the fundamental BCJ relation.

After the permutation $\tau$, $\alpha_1$ will offset the sub-currents $\mathcal{J}(A_I)$ in the left side of the equation and the sub-currents $\mathcal{J}(B_J)$ in the right by the $U(1)$-decoupling identity. Since $\alpha_1$ can only offset one sub-current at one time, when we consider the terms without sub-currents after the permutation $\tau$, we only need to consider the terms with at most one sub-current before the permutation $\tau$. And there are only three external lines in the sub-currents by the $U(1)$-decoupling identity. Thus we have the equation,
\begin{equation}\label{b1}
  \begin{split}
  &\Sl_{i=2}^{r-1}\Sl_{\tau_1\in OP(\{\alpha_1\}\bigcup\{\alpha_i,\alpha_{i+1}\})} \left(\frac{1}{2F^2}\right)^{\frac{r+s-3}{2}}S_{div\{\mathbf{\alpha_r},\mathbf{\beta_s}\}}\mathcal{J}(\alpha_{2})\cdots \mathcal{J}(\alpha_{1}\alpha_{i}\alpha_{i+1})\cdots \mathcal{J}(\alpha_r)\mathcal{J}(\beta_{1})\cdots \mathcal{J}(\beta_s)\\
  &+\Sl_{\tau_1}\left(\left(\frac{1}{2F^2}\right)^{\frac{r+s-1}{2}}S_{0}\mathcal{J}(\alpha_{1})\mathcal{J}(\alpha_2)\cdots \mathcal{J}(\alpha_r)\mathcal{J}(\beta_{1})\cdots \mathcal{J}(\beta_s)\right)\\
  =&\Sl_{j=1}^{s-1}\Sl_{\tau_2\in OP(\{\alpha_1\}\bigcup\{\beta_j,\beta_{j+1}\})} \left(\frac{1}{2F^2}\right)^{\frac{r+s-3}{2}}S_{div\{\mathbf{\alpha_r},\mathbf{\beta_s}\}}\mathcal{J}(\alpha_{2})\cdots \mathcal{J}(\alpha_r)\mathcal{J}(\beta_{1})\cdots \mathcal{J}(\alpha_1\beta_{j}\beta_{j+1})\cdots \mathcal{J}(\beta_s)\\
  &+\Sl_{\tau_2}\left(\left(\frac{1}{2F^2}\right)^{\frac{r+s-1}{2}}S_{0}\mathcal{J}(\alpha_{2})\cdots \mathcal{J}(\alpha_r)\mathcal{J}(\alpha_1)\mathcal{J}(\beta_{1})\cdots \mathcal{J}(\beta_s)\right).
  \end{split}
\end{equation}
Since the sub-currents are all the same, we can work with the coefficient matrices independently. Let us denote the corresponding part by,
\begin{equation}\label{}
  \boldsymbol{\mathrm{X}_{\alpha}^{(r-s)}}(r,s)+\boldsymbol{\mathrm{D}_{\alpha}^{(r-s)}}(r,s)
  =\boldsymbol{\mathrm{Y}_{\alpha}^{(r-s)}}(r,s)+\boldsymbol{\mathrm{Z}_{\alpha}^{(r-s)}}(r,s)
\end{equation}

What we can calculate directly is $\boldsymbol{\mathrm{X}_{\alpha}^{(r-s)}}(r,s)$, $\boldsymbol{\mathrm{Y}_{\alpha}^{(r-s)}}(r,s)$ and $\boldsymbol{\mathrm{Z}_{\alpha}^{(r-s)}}(r,s)$. Then we can obtain $\boldsymbol{\mathrm{D}_{\alpha}^{(r-s)}}(r,s)$, with an undecided first row. According to our recursion hypothesis, only when $r-s=-1,1,3$ both sides do not vanish. We will calculate the most difficult case $r-s=1$ in appendix. The calculation of $\boldsymbol{\mathrm{D}_{\beta}^{(r-s)}}(r,s)$ is almost the same, and we will also leave it in the appendix.

The form of the matrices $\boldsymbol{\mathrm{D}^{(r-s)}}(r,s)$ are given directly as following,
\begin{equation*}
\boldsymbol{\mathrm{D}_{\alpha}^{(-1)}}(r,s)=
\begin{pmatrix}
0 & 1 & 2 & \cdots & r-2 & r-1 & r\\
0 & r-1 & r & \cdots & r & r & r\\
0 & 0 & r-2 & \cdots & r & r & r\\
\vdots & \vdots & \vdots & \ddots & \vdots & \vdots & \vdots\\
0 & 0 & 0 & \cdots & 2 & r & r\\
0 & 0 & 0 & \cdots & 0 & 1 & r\\
\end{pmatrix},
\boldsymbol{\mathrm{D}_{\beta}^{(-1)}}(r,s)=
\begin{pmatrix}
s-1 & 1 & s & \cdots & s & s & s\\
s-2 & 0 & 2 & \cdots & s & s & s\\
\vdots & \vdots & \vdots & \ddots & \vdots & \vdots & \vdots\\
3 & 0 & 0 & \cdots & s-3 & s & s\\
2 & 0 & 0 & \cdots & 0 & s-2 & s\\
1 & 0 & 0 & \cdots & 0 & 0 & s-1\\
\end{pmatrix}.
\end{equation*}
\begin{equation*}
\boldsymbol{\mathrm{D}_{\alpha}^{(1)}}(r,s)=
\begin{pmatrix}
1 & 2 & \cdots & r-3 & r-2 & r-1\\
r-1 & r & \cdots & r & r & r\\
0 & r-2 & \cdots & r & r & r\\
\vdots & \vdots & \ddots & \vdots & \vdots & \vdots\\
0 & 0 & \cdots & 3 & r & r\\
0 & 0 & \cdots & 0 & 2 & r\\
0 & 0 & \cdots & 0 & 0 & 1\\
\end{pmatrix},
\boldsymbol{\mathrm{D}_{\beta}^{(1)}}(r,s)=
\begin{pmatrix}
s & s & s & \cdots & s & s\\
s-1 & 1 & s & \cdots & s & s\\
s-2 & 0 & 2 & \cdots & s & s\\
\vdots & \vdots & \vdots & \ddots & \vdots & \vdots\\
2 & 0 & 0 & \cdots & s-2 & s\\
1 & 0 & 0 & \cdots & 0 & s-1\\
0 & 0 & 0 & \cdots & 0 & 0\\
\end{pmatrix}.
\end{equation*}
\begin{equation}\label{}
  \boldsymbol{\mathrm{D}_{\alpha}^{(i)}}(r,s)=\mathbf{0},\ \boldsymbol{\mathrm{D}_{\beta}^{(i)}}(r,s)=\mathbf{0},\
  \mathrm{when}\ i\neq\pm 1.
\end{equation}
Now it is easy to verify that our form of $\mathbf{C}(r,s)$ satisfies all the equations, so it is a special solution which can be denoted as $c_s^{ij}$. However, we need to prove it is the unique solution. Since $\mathbf{R}$ is of rank $r-1$, (\ref{R}) restricts the $rs$-dimensional $\mathbf{C}(r,s)$ to a $s$-dimensional space. Furthermore, (\ref{S}) restricts the $s$-dimensional space to a $1$-dimensional space. So, we only have one parameter to adjust. What is it? It is nothing but the 'n' in (\ref{g}). So we have the solution for (\ref{R}) and (\ref{S})
\begin{equation}\label{}
  c^{ij}=c_g^{ij}+c_s^{ij}=n(-1)^{i+j}{r-1\choose i-1}{s-1\choose j-1}+c_s^{ij}.
\end{equation}

We need the following lemma to decide $n$:
\begin{lemma}\label{$c^{r1}=0$}
$c^{r1}=0$.
\end{lemma}
\begin{proof}
Consider the whole coefficient $p_1^2s_{\alpha_r\beta_1}$, expand $p_1^2$, so if there are not any diagrams and their sums can provide $s_{\alpha_r\beta_1}s_{\alpha_r\beta_1}$ without sub-currents, $c^{r1}$ is zero.
\begin{figure}[htbp]
  \centering
  \includegraphics[width=0.8\textwidth]{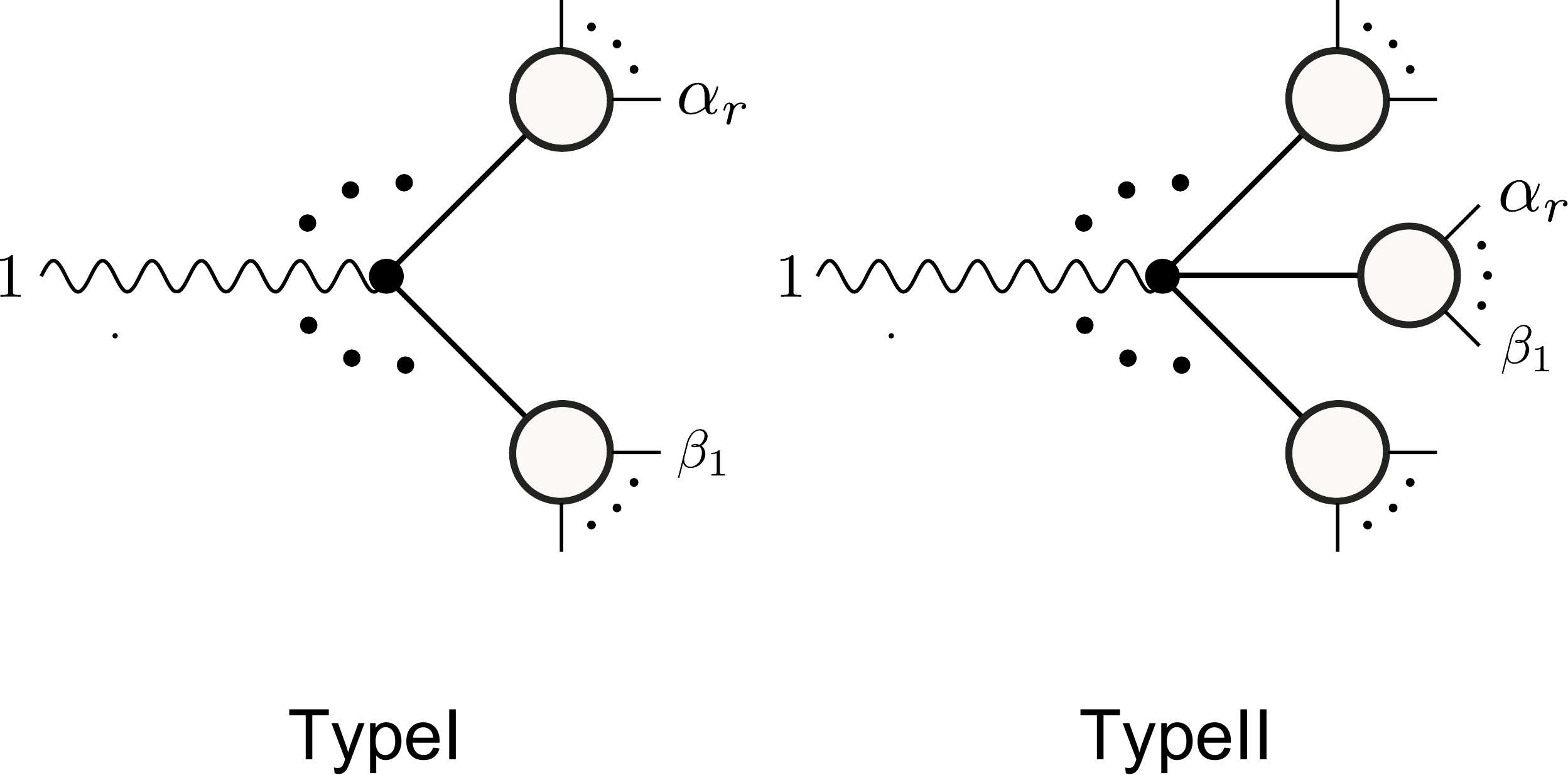}\\
  \caption{Two types of diagrams that may provide $s_{\alpha_r\beta_1}s_{\alpha_r\beta_1}$.}\label{310}
\end{figure}

There are two cases which may provide $s_{\alpha_r\beta_1}s_{\alpha_r\beta_1}$, the diagrams are given in Fig.\ref{310}. In the first one, $\alpha_r$ and $\beta_1$ are not in the same sub-current which is attached to $p_1$. In these diagrams, external line $\beta_1$ must be behind $\alpha_r$ in order to provide $s_{\alpha_r\beta_1}$. Thus the external lines or sub-currents containing them must be adjacent. However, according to the Feynman rule, the amplitudes do not contain momenta of adjacent lines attach to $p_1$ directly, so the other part of the whole coefficients can not provide $s_{\alpha_r\beta_1}$. The first case is killed.

The second case is that $\alpha_r$ and $\beta_1$ are in the same sub-currents. The amplitudes may be $(\cdots +p_{\alpha_r}+p_{\beta_1}+ \cdots)^2$ and provide $s_{\alpha_r\beta_1}$. However, according to the recursion hypothesis, we can not obtain $s_{\alpha_r\beta_1}$ because it does not have $\beta_j$ can make $\cdots \mathcal{J}(\beta_j)\mathcal{J}(\alpha_r)\mathcal{J}(\beta_1) \cdots$  or $\alpha_i$ can make $\cdots \mathcal{J}(\alpha_r)\mathcal{J}(\beta_1)\mathcal{J}(\alpha_i) \cdots$ in the rule of reading coefficients.

In one word, $s_{\alpha_r\beta_1}$ can appear either in the BCJ coefficient or in the vertex, but cannot in both. So we cannot obtain $s_{\alpha_r\beta_1}s_{\alpha_r\beta_1}$ in terms without sub-currents, which means $c^{r1}=0$.
\end{proof}
So
\begin{equation}\label{}
  c_g^{r1}=c^{r1}-c_s^{r1}=0,
\end{equation}
which means $n=0$. So $c_s^{ij}$ is just the unique solution.

%\paragraph{The coefficient matrix $\mathbf{\overline{C}}(r,r)$}\ \\

%Now let us consider the coefficient matrix $\mathbf{\overline{C}}(r,r)$. Since we have the generalized $U(1)$-decoupling identity, it is easy to get $\mathbf{\overline{C}}(r,r)$ from lemma\ref{$A_{RBCJ2}$},
%\begin{eqnarray}
%&&\bar{c}^{ij}=\left\{
%  \begin{array}{ll}
%  1,\ if\ i<j \\
%  0,\ otherwise
%  \end{array}
%  \right.,\ if\ r-s=\pm1, \\
%&&\bar{c}^{ij}=0,\ \ \ \ \ \ \ \ \ \ \ \ \ \ \ \ \ \ \ otherwise.
%\end{eqnarray}
\end{proof}
\section{Conclusions}

In this work, we proved the general BCJ relation in tree-level nonlinear sigma model with one external line off-shell under Cayley parametrization by proposing and proving its equivalent formula, the revised BCJ relation. The transform formula between these two relations are also given. The diagram representation of the result also gives a beautiful explanation of the previous results of KK relation in \cite{Chen:2013fya,Chen:2014dfa}. The permutation sum in the revised BCJ relation in nonlinear sigma model gives zero in the on-shell case, while in the off-shell case it is no longer zero. This off-shell permutation sum equals to a summation of sub-current products with the BCJ coefficients under a specific ordering. Besides, our proof is completely recursive, and thus has a minimal dependence of the character of nonlinear sigma model. Both the conciseness of the result and minimal model dependence indicate that this result may not be limited to nonlinear sigma model. For future work, the algebraic interpretation of these relations and the loop-level extensions are also deserved.

\appendix
\section{Calculation of the Matrices $\mathbf{D}^{(r-s)}(r,s)$}
In this part we will calculate the coefficient matrices of the term without sub-currents after $\alpha_1$ or $\beta_1$ is permuted among $\mathbf{\alpha_{r-1}}$ or $\mathbf{\beta_{s-1}}$ respectively. The basic idea of this proof is that we can view one amplitude by two methods, and thus we can obtain an equation to solve the part we want. The calculation is based on the result of the generalized $U(1)$-decoupling identity, which is proven in \cite{Chen:2014dfa}, and revised BCJ relation with less external lines, or the same external lines but less number of $\alpha$'s. A result which is frequently used in this proof is that the order of finite summations can be changed.

\subsection{$\alpha_1$ is permuted among $\mathbf{\alpha_{r-1}}$}
We discuss the case for $r-s=1$ in detail. Let us calculate $\boldsymbol{\mathrm{X}_{\alpha}^{(1)}}(r,s)$ first. Considering the $U(1)$-decoupling identity,
\begin{equation}\label{}
  \Sl_{\tau\in OP(\{\alpha_1\}\bigcup\{\alpha_i,\alpha_{i+1}\})} \mathcal{J}(\alpha_1\alpha_i\alpha_{i+1})=\frac{1}{2F^2}\mathcal{J}(\alpha_i)\mathcal{J}(\alpha_1)\mathcal{J}(\alpha_{i+1}),
\end{equation}
we have
\begin{equation}
\boldsymbol{\mathrm{X}_{\alpha}^{(1)}}(r,s)
=
\Sl_{k=2}^{r-1}
\begin{pmatrix}
* & * & \cdots & * & * & *\\
0 & 1 & \cdots & 1 & 1 & 1\\
0 & 0 & \cdots & 1 & 1 & 1\\
\vdots & \vdots & \ddots & \vdots & \vdots & \vdots\\
0 & 0 & \cdots & 1 & 1 & 1\\
0 & 0 & \cdots & 0 & 1 & 1\\
0 & 0 & \cdots & 0 & 0 & 1\\
\end{pmatrix}\\
=
\begin{pmatrix}
* & * & \cdots & * & * & *\\
0 & r-2 & \cdots & r-2 & r-2 & r-2\\
0 & 1 & \cdots & r-2 & r-2 & r-2\\
\vdots & \vdots & \ddots & \vdots & \vdots & \vdots\\
0 & 0 & \cdots & r-1 & r-2 & r-2\\
0 & 0 & \cdots & 0 & r-3 & r-2\\
0 & 0 & \cdots & 0 & 0 & r-2\\
\end{pmatrix}.
\end{equation}
where in the first matrices,
\begin{equation*}\label{}
  c^{(i+1)i}=\left\{
  \begin{array}{ll}
  0,\ if\ i<k \\
  1,\ if\ i\geqslant k
  \end{array}
  \right..
\end{equation*}

In a similar manner, we can calculate,
\begin{equation}
\boldsymbol{\mathrm{Y}_{\alpha}^{(1)}}(r,s)
=
\Sl_{k=1}^{s-1}
\begin{pmatrix}
* & * & \cdots & * & * & *\\
1 & 1 & \cdots & 1 & 1 & 1\\
0 & 1 & \cdots & 1 & 1 & 1\\
\vdots & \vdots & \ddots & \vdots & \vdots & \vdots\\
0 & 0 & \cdots & 0 & 1 & 1\\
0 & 0 & \cdots & 0 & 0 & 1\\
0 & 0 & \cdots & 0 & 0 & 0\\
\end{pmatrix}\\
=
\begin{pmatrix}
* & * & \cdots & * & * & *\\
r-2 & r-2 & \cdots & r-2 & r-2 & r-2\\
0 & r-3 & \cdots & r-2 & r-2 & r-2\\
\vdots & \vdots & \ddots & \vdots & \vdots & \vdots\\
0 & 0 & \cdots & 2 & r-2 & r-2\\
0 & 0 & \cdots & 0 & 1 & r-2\\
0 & 0 & \cdots & 0 & 0 & 0\\
\end{pmatrix}.
\end{equation}
where in the first matrices,
\begin{equation*}\label{}
  c^{(i+1)i}=\left\{
  \begin{array}{ll}
  1,\ if\ i\leqslant k \\
  0,\ if\ i>k
  \end{array}
  \right..
\end{equation*}

As for $\boldsymbol{\mathrm{Z}_{\alpha}^{(1)}}(r,s)$, we can calculate it directly,
\begin{equation}
\boldsymbol{\mathrm{Z}_{\alpha}^{(1)}}(r,s)
=
\Sl_{k=0}^{s}
\begin{pmatrix}
* & * & \cdots & * & * & *\\
0 & 1 & \cdots & 1 & 1 & 1\\
0 & 0 & \cdots & 1 & 1 & 1\\
\vdots & \vdots & \ddots & \vdots & \vdots & \vdots\\
0 & 0 & \cdots & 1 & 1 & 1\\
0 & 0 & \cdots & 0 & 1 & 1\\
0 & 0 & \cdots & 0 & 0 & 1\\
\end{pmatrix}\\
=
\begin{pmatrix}
* & * & \cdots & * & * & *\\
1 & r & \cdots & r & r & r\\
0 & 2 & \cdots & r & r & r\\
\vdots & \vdots & \ddots & \vdots & \vdots & \vdots\\
0 & 0 & \cdots & r-3 & r & r\\
0 & 0 & \cdots & 0 & r-2 & r\\
0 & 0 & \cdots & 0 & 0 & r-1\\
\end{pmatrix}.
\end{equation}
where in the first matrices,
\begin{equation*}\label{}
  c^{(i+1)i}=\left\{
  \begin{array}{ll}
  0,\ if\ i<k \\
  1,\ if\ i\geqslant k
  \end{array}
  \right..
\end{equation*}

Finally, we can get
\begin{equation}
\boldsymbol{\mathrm{D}_{\alpha}^{(1)}}(r,s)
=\boldsymbol{\mathrm{Y}_{\alpha}^{(1)}}(r,s)+\boldsymbol{\mathrm{Z}_{\alpha}^{(1)}}(r,s) -\boldsymbol{\mathrm{X}_{\alpha}^{(1)}}(r,s)\\
=\begin{pmatrix}
* & * & \cdots & * & * & *\\
r-1 & r & \cdots & r & r & r\\
0 & r-2 & \cdots & r & r & r\\
\vdots & \vdots & \ddots & \vdots & \vdots & \vdots\\
0 & 0 & \cdots & 3 & r & r\\
0 & 0 & \cdots & 0 & 2 & r\\
0 & 0 & \cdots & 0 & 0 & 1\\
\end{pmatrix},
\end{equation}
where the first row is undecided. The other two cases are similar.

\subsection{$\beta_1$ is permuted among $\mathbf{\beta_{s-1}}$}
The outline of this proof is that when a $\beta_1$ is permuted among $\mathbf{\beta_{s-1}}$, we can also think of it as the $\beta_1$ is permuted among a fixed order of $\mathbf{\alpha_r}\cup\mathbf{\beta_{s-1}}$, then sum up all the ordered permutation. We will calculate the terms without sub-currents from both points of view, compare them, and get the terms we want. In this part, we assume that $\tau_1\in OP(\beta_1\bigcup\mathbf{\beta_{s-1}})$ and $\tau_2\in OP(\{\beta_1\}\bigcup\{\mathbf{\alpha_r}\cup\mathbf{\beta_{s-1}}\})$ without other notations.

From the first point of view, which is the same as the running $\alpha$ case,
\begin{equation}\label{b1}
  \begin{split}
  &-\Sl_{\tau_1}\left(\Sl_{\sigma\in OP(\mathbf{\alpha_r}\bigcup\mathbf{\beta_s})}\left(\Sl_{i=1}^r\Sl_{\xi_{\sigma_k}>\xi_{\alpha_i}}s_{\alpha_i\sigma_k}\right)\mathcal{A}(1,\{\sigma\})\right)\\
  =&\Sl_{\tau_1}\left(\Sl_{div\{\mathbf{\alpha_r},\mathbf{\beta_s}\}}   \left(\frac{1}{2F^2}\right)^{\frac{R+S-1}{2}}S_{div\{\mathbf{\alpha_r},\mathbf{\beta_s}\}}\mathcal{J}(A_{1})\cdots \mathcal{J}(A_R)\mathcal{J}(B_{1})\cdots \mathcal{J}(B_S)\right).
  \end{split}
\end{equation}

From the second point of view, which is a little more tricky than the running $\alpha$ case, for each fixed order of $\mathbf{\alpha_r}\cup\mathbf{\beta_{s-1}}$, when $\beta_1$ is permuted among them, the coefficients are the same if we do not consider the coefficients  containing $\beta_1$(and fortunately we needn't consider it because it is not independent with other coefficients, as is seen in the transformation matrix $\mathbf{S}$). So we can calculate it with the $U(1)$-decoupling identity, then multiply it with the corresponding coefficients. Therefore we have
\begin{equation}\label{b1}
  \begin{split}
  &-\Sl_{\tau_1}\left(\Sl_{\sigma\in OP(\mathbf{\alpha_r}\bigcup\mathbf{\beta_s})}\left(\Sl_{i=1}^r\Sl_{\xi_{\sigma_k}>\xi_{\alpha_i}}s_{\alpha_i\sigma_k}\right)\mathcal{A}(1,\{\sigma\})\right)\\
  =&-\Sl_{\sigma\in OP(\mathbf{\alpha_r}\bigcup\mathbf{\beta_{s-1}})}\left(\left(\Sl_{i=1}^r\Sl_{\xi_{\sigma_k}>\xi_{\alpha_i}}s_{\alpha_i\sigma_k}\right)\Sl_{\tau_2}\mathcal{A}(1,\{\tau\})\right)
  \end{split}
\end{equation}
Using the result of the $U(1)$-decoupling identity, we can sum up the permutation $\tau_2$,
\begin{equation}\label{b2}
  =-\Sl_{\sigma\in OP(\mathbf{\alpha_r}\bigcup\mathbf{\beta_{s-1}})}\left(\left(\Sl_{i=1}^r\Sl_{\xi_{\sigma_k}>\xi_{\alpha_i}}s_{\alpha_i\sigma_k}\right)\frac{1}{2F^2} \Sl_{div\{\mathbf{\alpha_r}\cup\mathbf{\beta_{s-1}}\}\rightarrow\{T_{1}\},\{T_{2}\}}\mathcal{J}(\{T_{1}\})\mathcal{J}(\{T_{2}\})\right),
\end{equation}
Which is illustrated in Fig.\ref{fig_b_1}.
\begin{figure}
  \centering
  \includegraphics[width=10cm]{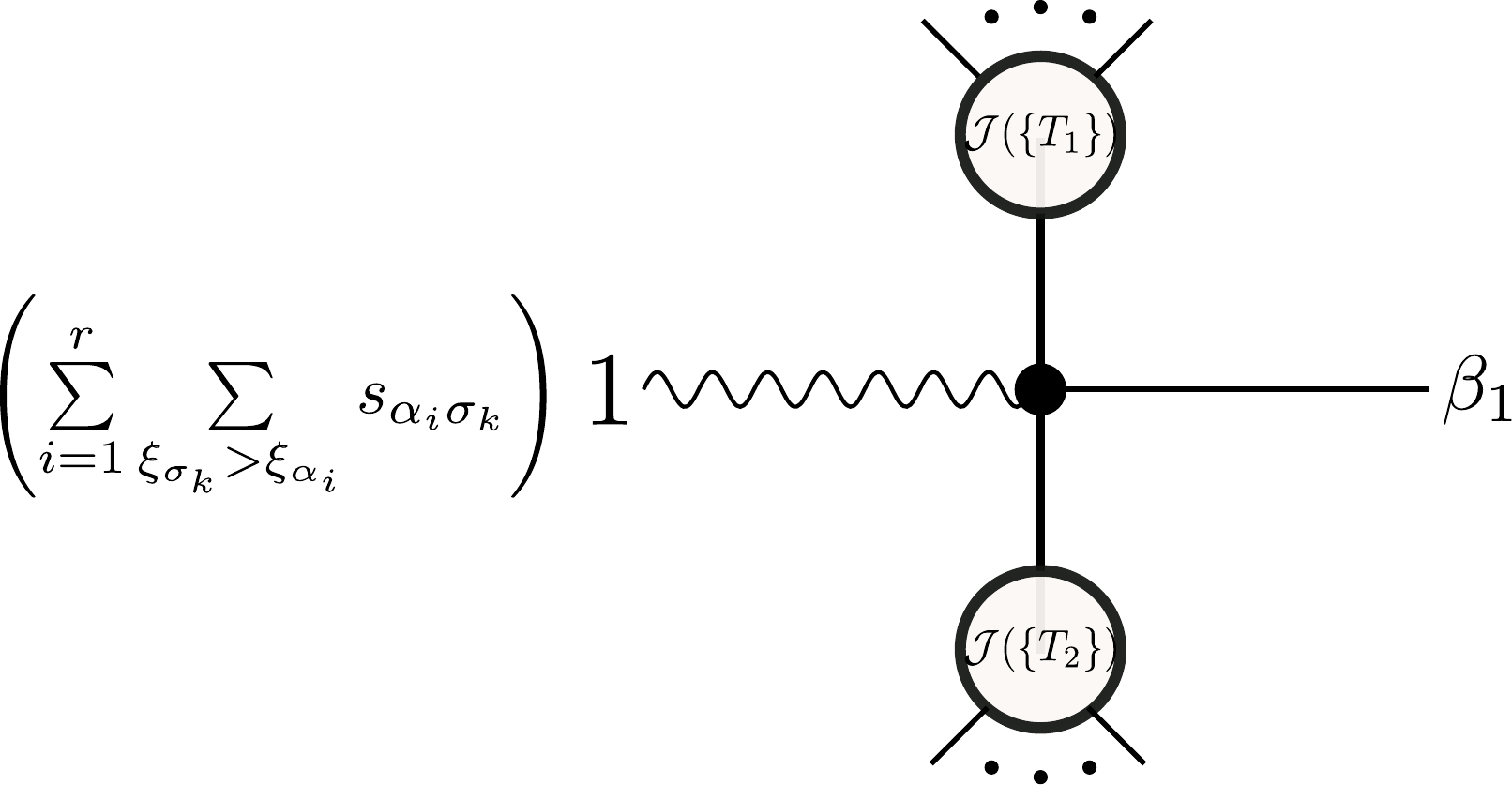}\\
  \caption{Schematic diagram of the result in $U(1)$-decoupling identity. After changing the order of summation, the summation of permutation of $\beta_1$ among $\mathbf{\alpha_r}\cup\mathbf{\beta_{s-1}}$ can be decided by the $U(1)$-decoupling identity.}\label{fig_b_1}
\end{figure}

Again, we can change the order of summation. However, $\mathbf{\alpha_r}\cup\mathbf{\beta_{s-1}}$ are divided into two parts $T_1$ and $T_2$, so $\sigma$ is composed of three permutations, namely permutation in $T_1$ ($\sigma_1$), permutation in $T_2$ ($\sigma_2$) and permutation between $T_1$ and $T_2$ ($\bar{\sigma}$). Meanwhile, we can divide the BCJ coefficients $\Sl_{i=1}^r\Sl_{\xi_{\sigma_k}>\xi_{\alpha_i}}s_{\alpha_i\sigma_k}$ into three parts as we do in (\ref{div_S}), $\Sl_{i=1}^r\Sl_{\xi_{\sigma_k}>\xi_{\alpha_i}}s_{\alpha_i\sigma_k}=S_1+S_2+\bar{S}$, where
\begin{equation}\label{b3}
  \begin{split}
  S_l=\Sl_{\alpha_i\in T_l}\Sl_{\substack{\sigma_k\in T_l\\\xi_{\sigma_k}>\xi_{\alpha_i}}}s_{\alpha_i\sigma_k},~\bar{S}=\Sl_{\alpha_i\in T_1}\Sl_{\sigma_k\in T_2}s_{\alpha_i\sigma_k}
  \end{split}
\end{equation}
Therefore,
\begin{equation}\label{b4}
  \begin{split}
  =&-\frac{1}{2F^2} \Sl_{div\{\mathbf{\alpha_r}\cup\mathbf{\beta_{s-1}}\}\rightarrow\{T_{1}\},\{T_{2}\}}\Sl_{\bar{\sigma}}\left(\left(\Sl_{\sigma_1}S_1\mathcal{J}(\{T_{1}\})\right)\left(\Sl_{\sigma_2}\mathcal{J}(\{T_{2}\})\right)\right)\\
  &-\frac{1}{2F^2} \Sl_{div\{\mathbf{\alpha_r}\cup\mathbf{\beta_{s-1}}\}\rightarrow\{T_{1}\},\{T_{2}\}}\Sl_{\bar{\sigma}}\left(\left(\Sl_{\sigma_1}\mathcal{J}(\{T_{1}\})\right)\left(\Sl_{\sigma_2}S_2\mathcal{J}(\{T_{2}\})\right)\right)\\
  &-\frac{1}{2F^2} \Sl_{div\{\mathbf{\alpha_r}\cup\mathbf{\beta_{s-1}}\}\rightarrow\{T_{1}\},\{T_{2}\}}\Sl_{\bar{\sigma}}\left(\bar{S}\left(\Sl_{\sigma_1}\mathcal{J}(\{T_{1}\})\right)\left(\Sl_{\sigma_2}\mathcal{J}(\{T_{2}\})\right)\right)\\
  \end{split}
\end{equation}
With the result of generalized $U(1)$-decoupling identity and the hypothesis of revised BCJ relation of less external lines, we can sum up all the $\sigma_1$ and $\sigma_2$. For example,
\begin{equation*}\label{}
  \Sl_{\sigma_1}\mathcal{J}(\{T_{1}\})=\Sl_{div\{\alpha_i,\beta_j\in T_1\}}   \left(\frac{1}{2F^2}\right)^{\frac{R_1+S_1-1}{2}}\mathcal{J}(A_{1}^{(T_1)})\cdots \mathcal{J}(A_{R_1}^{(T_1)})\mathcal{J}(B_{1}^{(T_1)})\cdots \mathcal{J}(B_{S_1}^{(T_1)})
\end{equation*}
\begin{equation*}\label{}
  \Sl_{\sigma_1}S_1\mathcal{J}(\{T_{1}\})
  =\Sl_{div\{\alpha_i,\beta_j\in T_1\}}   \left(\frac{1}{2F^2}\right)^{\frac{R_1+S_1-1}{2}}S_{div\{\alpha_i,\beta_j\in T_1\}}\mathcal{J}(A_{1}^{(T_1)})\cdots \mathcal{J}(A_{R_1}^{(T_1)})\mathcal{J}(B_{1}^{(T_1)})\cdots \mathcal{J}(B_{S_1}^{(T_1)})
\end{equation*}
where $R_1-S_1=\pm1$. Since they are of the same form except for the coefficients, we can combine them together,
\begin{equation}\label{b5}
  \begin{split}
  =&-\frac{1}{2F^2} \Sl_{div\{\mathbf{\alpha_r}\cup\mathbf{\beta_{s-1}}\}\rightarrow\{T_{1}\},\{T_{2}\}}\Sl_{\bar{\sigma}}\left(S_{div\{\alpha_i,\beta_j\in T_1\}}+S_{div\{\alpha_i,\beta_j\in T_2\}}+\bar{S}\right)\\
  &\left(\Sl_{div\{\alpha_i,\beta_j\in T_1\}}   \left(\frac{1}{2F^2}\right)^{\frac{R_1+S_1-1}{2}}\mathcal{J}(A_{1}^{(T_1)})\cdots \mathcal{J}(A_{R_1}^{(T_1)})\mathcal{J}(B_{1}^{(T_1)})\cdots \mathcal{J}(B_{S_1}^{(T_1)})\right)\\
  &\left(\Sl_{div\{\alpha_i,\beta_j\in T_2\}}   \left(\frac{1}{2F^2}\right)^{\frac{R_2+S_2-1}{2}}\mathcal{J}(A_{1}^{(T_2)})\cdots \mathcal{J}(A_{R_2}^{(T_2)})\mathcal{J}(B_{1}^{(T_2)})\cdots \mathcal{J}(B_{S_2}^{(T_2)})\right).\\
  \end{split}
\end{equation}

Now, we can get the equation,
\begin{equation}\label{b1}
  \begin{split}
  &\Sl_{\tau_1}\left(\Sl_{div\{\mathbf{\alpha_r},\mathbf{\beta_s}\}}   \left(\frac{1}{2F^2}\right)^{\frac{R+S-1}{2}}S_{div\{\mathbf{\alpha_r},\mathbf{\beta_s}\}}\mathcal{J}(A_{1})\cdots \mathcal{J}(A_R)\mathcal{J}(B_{1})\cdots \mathcal{J}(B_S)\right)\\
  =&-\frac{1}{2F^2} \Sl_{div\{\mathbf{\alpha_r}\cup\mathbf{\beta_{s-1}}\}\rightarrow\{T_{1}\},\{T_{2}\}}\Sl_{\bar{\sigma}}\left(S_{div\{\alpha_i,\beta_j\in T_1\}}+S_{div\{\alpha_i,\beta_j\in T_2\}}+\bar{S}\right)\\
  &\left(\Sl_{div\{\alpha_i,\beta_j\in T_1\}}   \left(\frac{1}{2F^2}\right)^{\frac{R_1+S_1-1}{2}}\mathcal{J}(A_{1}^{(T_1)})\cdots \mathcal{J}(A_{R_1}^{(T_1)})\mathcal{J}(B_{1}^{(T_1)})\cdots \mathcal{J}(B_{S_1}^{(T_1)})\right)\\
  &\left(\Sl_{div\{\alpha_i,\beta_j\in T_2\}}   \left(\frac{1}{2F^2}\right)^{\frac{R_2+S_2-1}{2}}\mathcal{J}(A_{1}^{(T_2)})\cdots \mathcal{J}(A_{R_2}^{(T_2)})\mathcal{J}(B_{1}^{(T_2)})\cdots \mathcal{J}(B_{S_2}^{(T_2)})\right)\\.
  \end{split}
\end{equation}
From which we can easily get the equation of terms without sub-currents,
\begin{equation}\label{b6}
  \begin{split}
  &\Sl_{j=2}^{s-1}\Sl_{\tau_1\in OP(\{\beta_1\}\bigcup\{\beta_j,\beta_{j+1}\})} \left(\frac{1}{2F^2}\right)^{\frac{r+s-3}{2}}S_{div\{\mathbf{\alpha_r},\mathbf{\beta_s}\}}\mathcal{J}(\alpha_{1})\cdots \mathcal{J}(\alpha_r)\mathcal{J}(\beta_{2})\cdots \mathcal{J}(\beta_1\beta_{j}\beta_{j+1})\cdots \mathcal{J}(\beta_s)\\
  &+\Sl_{\tau_1}\left(\left(\frac{1}{2F^2}\right)^{\frac{r+s-1}{2}}S_{0}\mathcal{J}(\alpha_{1})\cdots \mathcal{J}(\alpha_r)\mathcal{J}(\beta_{1})\cdots \mathcal{J}(\beta_s)\right)\\
  =&-\frac{1}{2F^2} \Sl_{div\{\mathbf{\alpha_r}\cup\mathbf{\beta_{s-1}}\}\rightarrow\{T_{1}\},\{T_{2}\}}\Sl_{\bar{\sigma}}\left(S_0^{(T_1)}+S_0^{(T_2)}+\bar{S}\right)\\
  &\left(\left(\frac{1}{2F^2}\right)^{\frac{r_1+s_1-1}{2}}\mathcal{J}(\alpha_{1}^{(T_1)})\cdots \mathcal{J}(\alpha_{r_1}^{(T_1)})\mathcal{J}(\beta_{1}^{(T_1)})\cdots \mathcal{J}(\beta_{s_1}^{(T_1)})\right)\\
  &\left(\left(\frac{1}{2F^2}\right)^{\frac{r_2+s_2-1}{2}}\mathcal{J}(\alpha_{1}^{(T_2)})\cdots \mathcal{J}(\alpha_{r_2}^{(T_2)})\mathcal{J}(\beta_{1}^{(T_2)})\cdots \mathcal{J}(\beta_{s_2}^{(T_2)})\right),\\
  \end{split}
\end{equation}
Since the sub-currents are all the same, we can work with the coefficient matrices independently. Let us denote the corresponding part by,
\begin{equation}\label{}
  \boldsymbol{\mathrm{X}_{\beta}^{(r-s)}}(r,s)+\boldsymbol{\mathrm{D}_{\beta}^{(r-s)}}(r,s)
  =\Sl_{r_1-s_1=\pm1}\boldsymbol{\mathrm{Y}_{\beta}^{(r_1-s_1,r_2-s_2)}}(r,s).
\end{equation}
where $r_1-s_1=\pm1$, $r_2-s_2=\pm1$. Thus considering $\beta_1$, the only diagrams whose term without sub-current do not vanish satisfy $r-s=-3,-1,1$.

We discuss the case $r-s=-3$ in detail. Let us calculate $\boldsymbol{\mathrm{X}_{\beta}^{(-3)}}(r,s)$ first. Considering the $U(1)$-decoupling identity,
\begin{equation}\label{}
  \Sl_{\tau\in OP(\{\beta_1\}\bigcup\{\beta_j,\beta_{j+1}\})} \mathcal{J}(\beta_1\beta_{j}\beta_{j+1})=\frac{1}{2F^2}\mathcal{J}(\beta_{j})\mathcal{J}(\beta_1)\mathcal{J}(\beta_{j+1}),
\end{equation}
we have
\begin{equation}
\mathbf{\mathrm{X}_{\beta}^{(-3)}}(r,s)
=
\Sl_{k=2}^{s-1}
\begin{pmatrix}
* & 0 & 1 & 1 & \cdots & 1 & 1 & 1\\
* & 0 & 0 & 1 & \cdots & 1 & 1 & 1\\
\vdots & \vdots & \vdots & \vdots & \ddots & \vdots & \vdots & \vdots\\
* & 0 & 0 & 0 & \cdots & 0 & 1 & 1\\
* & 0 & 0 & 0 & \cdots & 0 & 0 & 1\\
\end{pmatrix}\\
=
\begin{pmatrix}
* & 0 & s-3 & s-2 & \cdots & s-2 & s-2 & s-2\\
* & 0 & 0 & s-4 & \cdots & s-2 & s-2 & s-2\\
\vdots & \vdots & \vdots & \vdots & \ddots & \vdots & \vdots & \vdots\\
* & 0 & 0 & 0 & \cdots & 2 & s-2 & s-2\\
* & 0 & 0 & 0 & \cdots & 0 & 1 & s-2\\
\end{pmatrix}.
\end{equation}
where in the first matrices,
\begin{equation*}\label{}
  c^{(i-2)i}=\left\{
  \begin{array}{ll}
  1,\ if\ i\leqslant k \\
  0,\ if\ i>k
  \end{array}
  \right..
\end{equation*}

In this case, the only non-vanishing $\boldsymbol{\mathrm{Y}_{\beta}^{(r_1-s_1,r_2-s_2)}}(r,s)$ is $\boldsymbol{\mathrm{Y}_{\beta}^{(-1,-1)}}(r,s)$, so
\begin{equation}
\boldsymbol{\mathrm{Y}_{\beta}^{(-1,-1)}}(r,s)
=
\Sl_{k=2}^{s-1}
\begin{pmatrix}
* & 0 & 1 & 1 & \cdots & 1 & 1 & 1\\
* & 0 & 0 & 1 & \cdots & 1 & 1 & 1\\
\vdots & \vdots & \vdots & \vdots & \ddots & \vdots & \vdots & \vdots\\
* & 0 & 0 & 0 & \cdots & 0 & 1 & 1\\
* & 0 & 0 & 0 & \cdots & 0 & 0 & 1\\
\end{pmatrix}\\
=
\begin{pmatrix}
* & 0 & s-3 & s-2 & \cdots & s-2 & s-2 & s-2\\
* & 0 & 0 & s-4 & \cdots & s-2 & s-2 & s-2\\
\vdots & \vdots & \vdots & \vdots & \ddots & \vdots & \vdots & \vdots\\
* & 0 & 0 & 0 & \cdots & 2 & s-2 & s-2\\
* & 0 & 0 & 0 & \cdots & 0 & 1 & s-2\\
\end{pmatrix}.
\end{equation}
where in the first matrices,
\begin{equation*}\label{}
  c^{(i-2)i}=\left\{
  \begin{array}{ll}
  1,\ if\ i\leqslant k \\
  0,\ if\ i>k
  \end{array}
  \right..
\end{equation*}
Finally, we can get
\begin{equation}
\boldsymbol{\mathrm{D}_{\beta}^{(-3)}}(r,s)
=\boldsymbol{\mathrm{Y}_{\beta}^{(-1,-1)}}(r,s)-\boldsymbol{\mathrm{X}_{\beta}^{(-3)}}(r,s)\\
=\begin{pmatrix}
* & 0 & 0 & 0 & \cdots & 0 & 0 & 0\\
* & 0 & 0 & 0 & \cdots & 0 & 0 & 0\\
\vdots & \vdots & \vdots & \vdots & \ddots & \vdots & \vdots & \vdots\\
* & 0 & 0 & 0 & \cdots & 0 & 0 & 0\\
* & 0 & 0 & 0 & \cdots & 0 & 0 & 0\\
\end{pmatrix},
\end{equation}
where the first column is undecided. Other two cases are similar.

\newpage
\begin{figure}
  \centering
  \includegraphics[width=12cm]{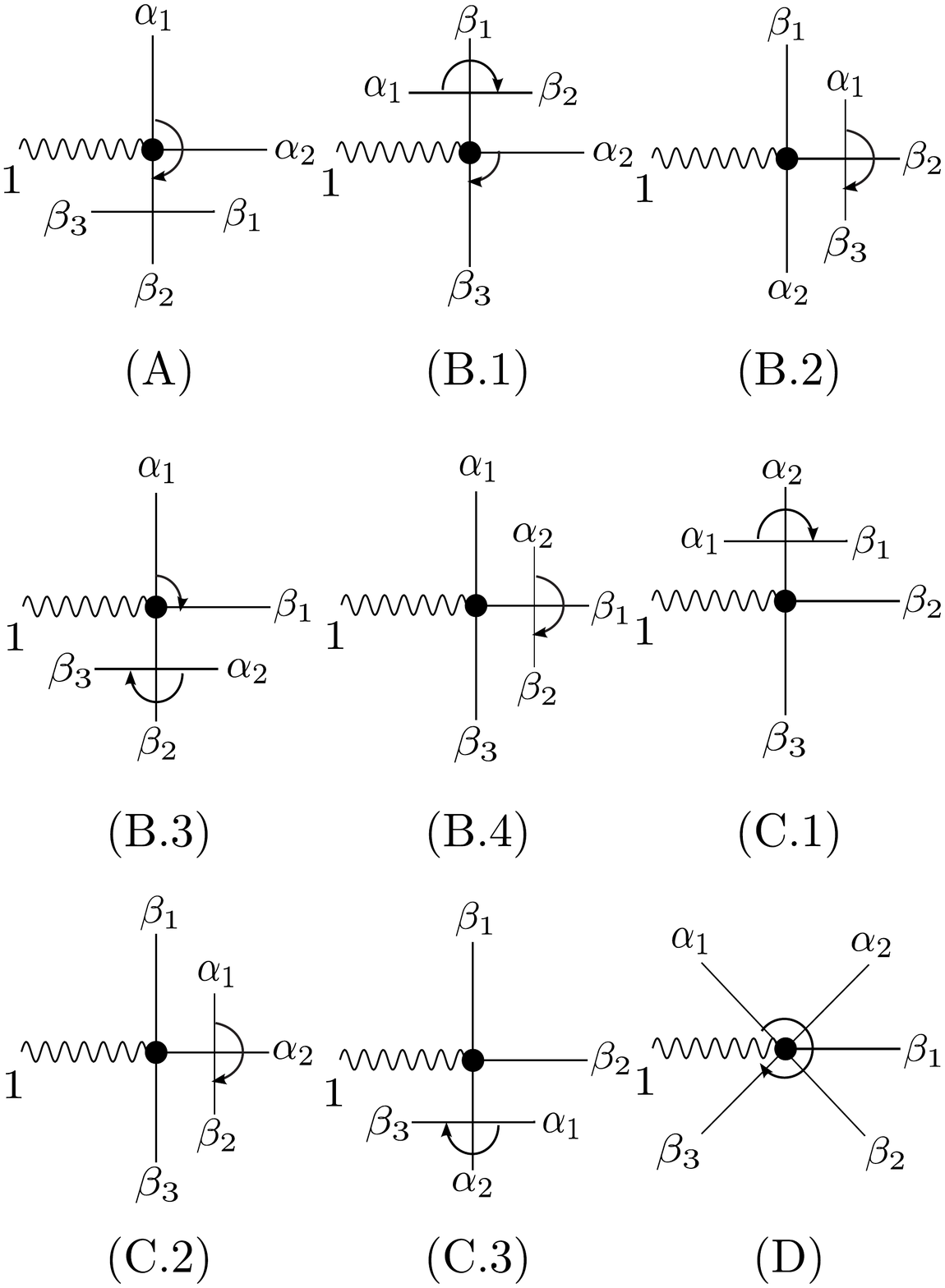}\\
  \caption{Diagrams of six-point currents, the arrow means permutation of $\alpha$ and $\beta$ at that vertex.}\label{example}
\end{figure}

\section*{Acknowledgments}
We thank  Y.J. Du for useful comments and kind suggestions.  This work has been supported by the Fundamental Research Funds for the Central Universities under contract 020414340080, NSF of China Grant under contract 11405084, the Open Project Program of State Key Laboratory of Theoretical Physics, Institute of Theoretical Physics, Chinese Academy of Sciences, China (No.Y5KF171CJ1) and the Jiangsu Ministry of Science and Technology under contract BK20131264.  We also thank Y. Gao, T. Han for hospitality and Key Laboratory of Theoretical Physics for hosting. 
\bibliographystyle{JHEP}
\addcontentsline{toc}{section}{References}
\bibliography{Refs}

\providecommand{\href}[2]{#2}\begingroup\raggedright\begin{thebibliography}{10}

\bibitem{Bern:2008qj}
Z.~Bern, J.~Carrasco and H.~Johansson, \emph{New relations for gauge-theory
  amplitudes}, \href{http://dx.doi.org/10.1103/PhysRevD.78.085011}{\emph{Phys.
  Rev. D} {\bf 78} (2008) 085011}, [\href{http://arxiv.org/abs/0805.3993}{{\tt
  0805.3993}}].

\bibitem{Kleiss:1988ne}
R.~Kleiss and H.~Kuijf, \emph{Multi - gluon cross-sections and five jet
  production at hadron colliders},
  \href{http://dx.doi.org/10.1016/0550-3213(89)90574-9}{\emph{Nucl. Phys. B}
  {\bf 312} (1989) 616--644}.

\bibitem{BjerrumBohr:2009rd}
N.~Bjerrum-Bohr, P.~H. Damgaard and P.~Vanhove, \emph{{Minimal Basis for Gauge
  Theory Amplitudes}},
  \href{http://dx.doi.org/10.1103/PhysRevLett.103.161602}{\emph{Phys.Rev.Lett.}
  {\bf 103} (2009) 161602}, [\href{http://arxiv.org/abs/0907.1425}{{\tt
  0907.1425}}].

\bibitem{Stieberger:2009hq}
S.~Stieberger, \emph{{Open $\&$ Closed vs. Pure Open String Disk Amplitudes}},
  \href{http://arxiv.org/abs/0907.2211}{{\tt 0907.2211}}.

\bibitem{DelDuca:1999rs}
V.~Del~Duca, L.~J. Dixon and F.~Maltoni, \emph{New color decompositions for
  gauge amplitudes at tree and loop level},
  \href{http://dx.doi.org/10.1016/S0550-3213(99)00809-3}{\emph{Nucl. Phys. B}
  {\bf 571} (2000) 51--70}, [\href{http://arxiv.org/abs/hep-ph/9910563}{{\tt
  hep-ph/9910563}}].

\bibitem{Britto:2004ap}
R.~Britto, F.~Cachazo and B.~Feng, \emph{New recursion relations for tree
  amplitudes of gluons},
  \href{http://dx.doi.org/10.1016/j.nuclphysb.2005.02.030}{\emph{Nucl. Phys. B}
  {\bf 715} (2005) 499--522}, [\href{http://arxiv.org/abs/hep-th/0412308}{{\tt
  hep-th/0412308}}].

\bibitem{Britto:2005fq}
R.~Britto, F.~Cachazo, B.~Feng and E.~Witten, \emph{{Direct proof of tree-level
  recursion relation in Yang-Mills theory}},
  \href{http://dx.doi.org/10.1103/PhysRevLett.94.181602}{\emph{Phys.Rev.Lett.}
  {\bf 94} (2005) 181602}, [\href{http://arxiv.org/abs/hep-th/0501052}{{\tt
  hep-th/0501052}}].

\bibitem{Feng:2010my}
B.~Feng, R.~Huang and Y.~Jia, \emph{{Gauge Amplitude Identities by On-shell
  Recursion Relation in S-matrix Program}},
  \href{http://dx.doi.org/10.1016/j.physletb.2010.11.011}{\emph{Phys.Lett.}
  {\bf B695} (2011) 350--353}, [\href{http://arxiv.org/abs/1004.3417}{{\tt
  1004.3417}}].

\bibitem{Tye:2010kg}
H.~Tye and Y.~Zhang, \emph{{Remarks on the identities of gluon tree
  amplitudes}},
  \href{http://dx.doi.org/10.1103/PhysRevD.82.087702}{\emph{Phys.Rev.} {\bf
  D82} (2010) 087702}, [\href{http://arxiv.org/abs/1007.0597}{{\tt
  1007.0597}}].

\bibitem{Cachazo:2012uq}
F.~Cachazo, \emph{{Fundamental BCJ Relation in N=4 SYM From The Connected
  Formulation}},  \href{http://arxiv.org/abs/1206.5970}{{\tt 1206.5970}}.

\bibitem{Chen:2011jxa}
Y.-X. Chen, Y.-J. Du and B.~Feng, \emph{{A Proof of the Explicit Minimal-basis
  Expansion of Tree Amplitudes in Gauge Field Theory}},
  \href{http://dx.doi.org/10.1007/JHEP02(2011)112}{\emph{JHEP} {\bf 1102}
  (2011) 112}, [\href{http://arxiv.org/abs/1101.0009}{{\tt 1101.0009}}].

\bibitem{Mafra:2011kj}
C.~R. Mafra, O.~Schlotterer and S.~Stieberger, \emph{{Explicit BCJ Numerators
  from Pure Spinors}},
  \href{http://dx.doi.org/10.1007/JHEP07(2011)092}{\emph{JHEP} {\bf 1107}
  (2011) 092}, [\href{http://arxiv.org/abs/1104.5224}{{\tt 1104.5224}}].

\bibitem{Monteiro:2011pc}
R.~Monteiro and D.~O'Connell, \emph{{The Kinematic Algebra From the Self-Dual
  Sector}}, \href{http://dx.doi.org/10.1007/JHEP07(2011)007}{\emph{JHEP} {\bf
  1107} (2011) 007}, [\href{http://arxiv.org/abs/1105.2565}{{\tt 1105.2565}}].

\bibitem{BjerrumBohr:2012mg}
N.~Bjerrum-Bohr, P.~H. Damgaard, R.~Monteiro and D.~O'Connell, \emph{{Algebras
  for Amplitudes}},
  \href{http://dx.doi.org/10.1007/JHEP06(2012)061}{\emph{JHEP} {\bf 1206}
  (2012) 061}, [\href{http://arxiv.org/abs/1203.0944}{{\tt 1203.0944}}].

\bibitem{Fu:2012uy}
C.-H. Fu, Y.-J. Du and B.~Feng, \emph{{An algebraic approach to BCJ
  numerators}}, \href{http://dx.doi.org/10.1007/JHEP03(2013)050}{\emph{JHEP}
  {\bf 1303} (2013) 050}, [\href{http://arxiv.org/abs/1212.6168}{{\tt
  1212.6168}}].

\bibitem{Bern:2011ia}
Z.~Bern and T.~Dennen, \emph{{A Color Dual Form for Gauge-Theory Amplitudes}},
  \href{http://dx.doi.org/10.1103/PhysRevLett.107.081601}{\emph{Phys.Rev.Lett.}
  {\bf 107} (2011) 081601}, [\href{http://arxiv.org/abs/1103.0312}{{\tt
  1103.0312}}].

\bibitem{Du:2013sha}
Y.-J. Du, B.~Feng and C.-H. Fu, \emph{{The Construction of Dual-trace Factor in
  Yang-Mills Theory}},
  \href{http://dx.doi.org/10.1007/JHEP07(2013)057}{\emph{JHEP} {\bf 1307}
  (2013) 057}, [\href{http://arxiv.org/abs/1304.2978}{{\tt 1304.2978}}].

\bibitem{Fu:2013qna}
C.-H. Fu, Y.-J. Du and B.~Feng, \emph{{Note on Construction of Dual-trace
  Factor in Yang-Mills Theory}},
  \href{http://dx.doi.org/10.1007/JHEP10(2013)069}{\emph{JHEP} {\bf 1310}
  (2013) 069}, [\href{http://arxiv.org/abs/1305.2996}{{\tt 1305.2996}}].

\bibitem{Broedel:2011pd}
J.~Broedel and J.~J.~M. Carrasco, \emph{{Virtuous Trees at Five and Six Points
  for Yang-Mills and Gravity}},
  \href{http://dx.doi.org/10.1103/PhysRevD.84.085009}{\emph{Phys.Rev.} {\bf
  D84} (2011) 085009}, [\href{http://arxiv.org/abs/1107.4802}{{\tt
  1107.4802}}].

\bibitem{Fu:2014pya}
C.-H. Fu, Y.-J. Du and B.~Feng, \emph{{Note on symmetric BCJ numerator}},
  \href{http://dx.doi.org/10.1007/JHEP08(2014)098}{\emph{JHEP} {\bf 1408}
  (2014) 098}, [\href{http://arxiv.org/abs/1403.6262}{{\tt 1403.6262}}].

\bibitem{Naculich:2014rta}
S.~G. Naculich, \emph{{Scattering equations and virtuous kinematic numerators
  and dual-trace functions}},
  \href{http://dx.doi.org/10.1007/JHEP07(2014)143}{\emph{JHEP} {\bf 1407}
  (2014) 143}, [\href{http://arxiv.org/abs/1404.7141}{{\tt 1404.7141}}].

\bibitem{Cachazo:2013gna}
F.~Cachazo, S.~He and E.~Y. Yuan, \emph{{Scattering Equations and KLT
  Orthogonality}},
  \href{http://dx.doi.org/10.1103/PhysRevD.90.065001}{\emph{Phys.Rev.} {\bf
  D90} (2014) 065001}, [\href{http://arxiv.org/abs/1306.6575}{{\tt
  1306.6575}}].

\bibitem{Cachazo:2013hca}
F.~Cachazo, S.~He and E.~Y. Yuan, \emph{{Scattering of Massless Particles in
  Arbitrary Dimensions}},
  \href{http://dx.doi.org/10.1103/PhysRevLett.113.171601}{\emph{Phys.Rev.Lett.}
  {\bf 113} (2014) 171601}, [\href{http://arxiv.org/abs/1307.2199}{{\tt
  1307.2199}}].

\bibitem{Cachazo:2013iea}
F.~Cachazo, S.~He and E.~Y. Yuan, \emph{{Scattering of Massless Particles:
  Scalars, Gluons and Gravitons}},
  \href{http://dx.doi.org/10.1007/JHEP07(2014)033}{\emph{JHEP} {\bf 1407}
  (2014) 033}, [\href{http://arxiv.org/abs/1309.0885}{{\tt 1309.0885}}].

\bibitem{Naculich:2014naa}
S.~G. Naculich, \emph{{Scattering equations and BCJ relations for gauge and
  gravitational amplitudes with massive scalar particles}},
  \href{http://dx.doi.org/10.1007/JHEP09(2014)029}{\emph{JHEP} {\bf 1409}
  (2014) 029}, [\href{http://arxiv.org/abs/1407.7836}{{\tt 1407.7836}}].

\bibitem{Bern:2010yg}
Z.~Bern, T.~Dennen, Y.-t. Huang and M.~Kiermaier, \emph{{Gravity as the Square
  of Gauge Theory}},
  \href{http://dx.doi.org/10.1103/PhysRevD.82.065003}{\emph{Phys.Rev.} {\bf
  D82} (2010) 065003}, [\href{http://arxiv.org/abs/1004.0693}{{\tt
  1004.0693}}].

\bibitem{BjerrumBohr:2011xe}
N.~Bjerrum-Bohr, P.~Damgaard, H.~Johansson and T.~Sondergaard,
  \emph{{Monodromy--like Relations for Finite Loop Amplitudes}},
  \href{http://dx.doi.org/10.1007/JHEP05(2011)039}{\emph{JHEP} {\bf 1105}
  (2011) 039}, [\href{http://arxiv.org/abs/1103.6190}{{\tt 1103.6190}}].

\bibitem{Carrasco:2011mn}
J.~J. Carrasco and H.~Johansson, \emph{{Five-Point Amplitudes in N=4
  Super-Yang-Mills Theory and N=8 Supergravity}},
  \href{http://dx.doi.org/10.1103/PhysRevD.85.025006}{\emph{Phys.Rev.} {\bf
  D85} (2012) 025006}, [\href{http://arxiv.org/abs/1106.4711}{{\tt
  1106.4711}}].

\bibitem{Boels:2011tp}
R.~H. Boels and R.~S. Isermann, \emph{{New relations for scattering amplitudes
  in Yang-Mills theory at loop level}},
  \href{http://dx.doi.org/10.1103/PhysRevD.85.021701}{\emph{Phys.Rev.} {\bf
  D85} (2012) 021701}, [\href{http://arxiv.org/abs/1109.5888}{{\tt
  1109.5888}}].

\bibitem{Boels:2011mn}
R.~H. Boels and R.~S. Isermann, \emph{{Yang-Mills amplitude relations at loop
  level from non-adjacent BCFW shifts}},
  \href{http://dx.doi.org/10.1007/JHEP03(2012)051}{\emph{JHEP} {\bf 1203}
  (2012) 051}, [\href{http://arxiv.org/abs/1110.4462}{{\tt 1110.4462}}].

\bibitem{Carrasco:2012ca}
J.~J.~M. Carrasco, M.~Chiodaroli, M.~G{\"u}naydin and R.~Roiban,
  \emph{{One-loop four-point amplitudes in pure and matter-coupled N $\leq$ 4
  supergravity}}, \href{http://dx.doi.org/10.1007/JHEP03(2013)056}{\emph{JHEP}
  {\bf 1303} (2013) 056}, [\href{http://arxiv.org/abs/1212.1146}{{\tt
  1212.1146}}].

\bibitem{Bern:2012uf}
Z.~Bern, J.~Carrasco, L.~Dixon, H.~Johansson and R.~Roiban, \emph{{Simplifying
  Multiloop Integrands and Ultraviolet Divergences of Gauge Theory and Gravity
  Amplitudes}},
  \href{http://dx.doi.org/10.1103/PhysRevD.85.105014}{\emph{Phys.Rev.} {\bf
  D85} (2012) 105014}, [\href{http://arxiv.org/abs/1201.5366}{{\tt
  1201.5366}}].

\bibitem{Du:2012mt}
Y.-J. Du and H.~Luo, \emph{{On General BCJ Relation at One-loop Level in
  Yang-Mills Theory}},
  \href{http://dx.doi.org/10.1007/JHEP01(2013)129}{\emph{JHEP} {\bf 1301}
  (2013) 129}, [\href{http://arxiv.org/abs/1207.4549}{{\tt 1207.4549}}].

\bibitem{Oxburgh:2012zr}
S.~Oxburgh and C.~White, \emph{{BCJ duality and the double copy in the soft
  limit}}, \href{http://dx.doi.org/10.1007/JHEP02(2013)127}{\emph{JHEP} {\bf
  1302} (2013) 127}, [\href{http://arxiv.org/abs/1210.1110}{{\tt 1210.1110}}].

\bibitem{Saotome:2012vy}
R.~Saotome and R.~Akhoury, \emph{{Relationship Between Gravity and Gauge
  Scattering in the High Energy Limit}},
  \href{http://dx.doi.org/10.1007/JHEP01(2013)123}{\emph{JHEP} {\bf 1301}
  (2013) 123}, [\href{http://arxiv.org/abs/1210.8111}{{\tt 1210.8111}}].

\bibitem{Boels:2012ew}
R.~H. Boels, B.~A. Kniehl, O.~V. Tarasov and G.~Yang, \emph{{Color-kinematic
  Duality for Form Factors}},
  \href{http://dx.doi.org/10.1007/JHEP02(2013)063}{\emph{JHEP} {\bf 1302}
  (2013) 063}, [\href{http://arxiv.org/abs/1211.7028}{{\tt 1211.7028}}].

\bibitem{Boels:2013bi}
R.~H. Boels, R.~S. Isermann, R.~Monteiro and D.~O'Connell,
  \emph{{Colour-Kinematics Duality for One-Loop Rational Amplitudes}},
  \href{http://dx.doi.org/10.1007/JHEP04(2013)107}{\emph{JHEP} {\bf 1304}
  (2013) 107}, [\href{http://arxiv.org/abs/1301.4165}{{\tt 1301.4165}}].

\bibitem{Bjerrum-Bohr:2013iza}
N.~E.~J. Bjerrum-Bohr, T.~Dennen, R.~Monteiro and D.~O'Connell,
  \emph{{Integrand Oxidation and One-Loop Colour-Dual Numerators in N=4 Gauge
  Theory}}, \href{http://dx.doi.org/10.1007/JHEP07(2013)092}{\emph{JHEP} {\bf
  1307} (2013) 092}, [\href{http://arxiv.org/abs/1303.2913}{{\tt 1303.2913}}].

\bibitem{Bern:2013yya}
Z.~Bern, S.~Davies, T.~Dennen, Y.-t. Huang and J.~Nohle,
  \emph{{Color-Kinematics Duality for Pure Yang-Mills and Gravity at One and
  Two Loops}},  \href{http://arxiv.org/abs/1303.6605}{{\tt 1303.6605}}.

\bibitem{Nohle:2013bfa}
J.~Nohle, \emph{{Color-Kinematics Duality in One-Loop Four-Gluon Amplitudes
  with Matter}},
  \href{http://dx.doi.org/10.1103/PhysRevD.90.025020}{\emph{Phys.Rev.} {\bf
  D90} (2014) 025020}, [\href{http://arxiv.org/abs/1309.7416}{{\tt
  1309.7416}}].

\bibitem{Sondergaard:2009za}
T.~Sondergaard, \emph{{New Relations for Gauge-Theory Amplitudes with Matter}},
  \href{http://dx.doi.org/10.1016/j.nuclphysb.2009.07.002}{\emph{Nucl.Phys.}
  {\bf B821} (2009) 417--430}, [\href{http://arxiv.org/abs/0903.5453}{{\tt
  0903.5453}}].

\bibitem{Jia:2010nz}
Y.~Jia, R.~Huang and C.-Y. Liu, \emph{{$U(1)$-decoupling, KK and BCJ relations
  in $\mathcal{N}=4$ SYM}},
  \href{http://dx.doi.org/10.1103/PhysRevD.82.065001}{\emph{Phys. Rev.} {\bf
  D82} (2010) 065001}, [\href{http://arxiv.org/abs/1005.1821}{{\tt
  1005.1821}}].

\bibitem{Du:2011js}
Y.-J. Du, B.~Feng and C.-H. Fu, \emph{{BCJ Relation of Color Scalar Theory and
  KLT Relation of Gauge Theory}},
  \href{http://dx.doi.org/10.1007/JHEP08(2011)129}{\emph{JHEP} {\bf 1108}
  (2011) 129}, [\href{http://arxiv.org/abs/1105.3503}{{\tt 1105.3503}}].

\bibitem{Bargheer:2012gv}
T.~Bargheer, S.~He and T.~McLoughlin, \emph{{New Relations for
  Three-Dimensional Supersymmetric Scattering Amplitudes}},
  \href{http://dx.doi.org/10.1103/PhysRevLett.108.231601}{\emph{Phys.Rev.Lett.}
  {\bf 108} (2012) 231601}, [\href{http://arxiv.org/abs/1203.0562}{{\tt
  1203.0562}}].

\bibitem{Kampf:2012fn}
K.~Kampf, J.~Novotny and J.~Trnka, \emph{{Recursion Relations for Tree-level
  Amplitudes in the SU(N) Non-linear Sigma Model}},
  \href{http://dx.doi.org/10.1103/PhysRevD.87.081701}{\emph{Phys.Rev.} {\bf
  D87} (2013) 081701}, [\href{http://arxiv.org/abs/1212.5224}{{\tt
  1212.5224}}].

\bibitem{Kampf:2013vha}
K.~Kampf, J.~Novotny and J.~Trnka, \emph{{Tree-level Amplitudes in the
  Nonlinear Sigma Model}},
  \href{http://dx.doi.org/10.1007/JHEP05(2013)032}{\emph{JHEP} {\bf 1305}
  (2013) 032}, [\href{http://arxiv.org/abs/1304.3048}{{\tt 1304.3048}}].

\bibitem{Chen:2013fya}
G.~Chen and Y.-J. Du, \emph{{Amplitude Relations in Non-linear Sigma Model}},
  \href{http://dx.doi.org/10.1007/JHEP01(2014)061}{\emph{JHEP} {\bf 1401}
  (2014) 061}, [\href{http://arxiv.org/abs/1311.1133}{{\tt 1311.1133}}].

\bibitem{Cachazo:2014xea}
F.~Cachazo, S.~He and E.~Y. Yuan, \emph{{Scattering Equations and Matrices:
  From Einstein To Yang-Mills, DBI and NLSM}},
  \href{http://dx.doi.org/10.1007/JHEP07(2015)149}{\emph{JHEP} {\bf 07} (2015)
  149}, [\href{http://arxiv.org/abs/1412.3479}{{\tt 1412.3479}}].

\bibitem{Chen:2014dfa}
G.~Chen, Y.-J. Du, S.~Li and H.~Liu, \emph{{Note on off-shell relations in
  nonlinear sigma model}},
  \href{http://dx.doi.org/10.1007/JHEP03(2015)156}{\emph{JHEP} {\bf 03} (2015)
  156}, [\href{http://arxiv.org/abs/1412.3722}{{\tt 1412.3722}}].

\bibitem{Du:2015esa}
Y.-J. Du and H.~Luo, \emph{{On single and double soft behaviors in NLSM}},
  \href{http://dx.doi.org/10.1007/JHEP08(2015)058}{\emph{JHEP} {\bf 08} (2015)
  058}, [\href{http://arxiv.org/abs/1505.04411}{{\tt 1505.04411}}].

\bibitem{Carrasco:2016ldy}
J.~J.~M. Carrasco, C.~R. Mafra and O.~Schlotterer, \emph{{Abelian Z-theory:
  NLSM amplitudes and alpha'-corrections from the open string}},
  \href{http://arxiv.org/abs/1608.02569}{{\tt 1608.02569}}.

\bibitem{Du:2016tbc}
Y.-J. Du and C.-H. Fu, \emph{{Explicit BCJ numerators of nonlinear sigma
  model}},  \href{http://arxiv.org/abs/1606.05846}{{\tt 1606.05846}}.

\bibitem{Ma:2011um}
Q.~Ma, Y.-J. Du and Y.-X. Chen, \emph{{On Primary Relations at Tree-level in
  String Theory and Field Theory}},
  \href{http://dx.doi.org/10.1007/JHEP02(2012)061}{\emph{JHEP} {\bf 1202}
  (2012) 061}, [\href{http://arxiv.org/abs/1109.0685}{{\tt 1109.0685}}].

\end{thebibliography}\endgroup


\providecommand{\href}[2]{#2}\begingroup\raggedright\endgroup

\end{document}